\documentclass[12pt, draftclsnofoot, onecolumn]{IEEEtran}
\usepackage{graphicx}
\usepackage{amsthm}
\usepackage{epsfig}
\usepackage{latexsym}
\usepackage{amsfonts}
\usepackage{here}
\usepackage{rawfonts}
\usepackage[latin1]{inputenc}
\usepackage[T1]{fontenc}
\usepackage{calc}
\usepackage{capitalgreekitalic}
\usepackage{url}
\usepackage{enumerate}
\usepackage{color}
\usepackage[tbtags]{amsmath}
\usepackage{amssymb}
\usepackage{upref}
\usepackage{epic,eepic}
\usepackage{times}
\usepackage{dsfont}
\usepackage{comment}
\usepackage{cite}
\usepackage{bbm} 

\newtheorem{theorem}{\bf Theorem}
\newtheorem{proposition}{\bf Proposition}
\newtheorem{lemma}{\bf Lemma}

\usepackage{dsfont}

\newcounter{step}
\newlength{\totlinewidth}
\newenvironment{algorithm}{%
  \rule{\linewidth}{1pt}
  \begin{list}{}%
    {\usecounter{step}%
      \settowidth{\labelwidth}{\textbf{Step 2:}}%
      \setlength{\leftmargin}{\labelwidth}%
      \setlength{\topsep}{-2pt}%
      \addtolength{\leftmargin}{\labelsep}%
      \addtolength{\leftmargin}{2mm}%
      \setlength{\rightmargin}{2mm}%
      \setlength{\totlinewidth}{\linewidth}%
      \addtolength{\totlinewidth}{\leftmargin}%
      \addtolength{\totlinewidth}{\rightmargin}%
      \setlength{\parsep}{0mm}%
      \raggedright}}%
  {\end{list}%
  \rule{\linewidth}{1pt}}
\newcounter{substep}

  {\end{list}}

\newlength{\aligntop}
\newlength{\alignbot}
 \makeatletter
\renewenvironment{align}{%
  \vspace{\aligntop}
  \start@align\@ne\st@rredfalse\m@ne
}{%
  \math@cr \black@\totwidth@
  \egroup
  \ifingather@
    \restorealignstate@
    \egroup
    \nonumber
    \ifnum0=`{\fi\iffalse}\fi
  \else
    $$%
  \fi
  \ignorespacesafterend%
  \vspace{\alignbot}\par\noindent
} \makeatother

\IEEEoverridecommandlockouts

\usepackage{algorithm}
\usepackage{algorithmic}

\makeatletter
\newcommand\semihuge{\@setfontsize\semihuge{19.3}{25}}
\makeatother

\makeatletter
\newcommand\semismall{\@setfontsize\semihuge{12.4}{15}}
\makeatother

\usepackage{subfigure}

\begin{document}

\title{\huge A Joint Learning and Communications Framework for Federated Learning over Wireless Networks}

\author{{Mingzhe Chen}, \emph{Member, IEEE}, Zhaohui Yang, \emph{Member, IEEE}, Walid Saad, \emph{Fellow, IEEE}, Changchuan Yin, \emph{Senior Member, IEEE}, H. Vincent Poor, \emph{Life Fellow, IEEE}, and Shuguang Cui, \emph{Fellow, IEEE} \vspace*{-2.2em}\\ 
\thanks{M. Chen and H. V. Poor are with the Department of Electrical and Computer Engineering, Princeton University, Princeton, NJ, 08544, USA, Emails: \protect{mingzhec@princeton.edu}, \protect{poor@princeton.edu}.}
\thanks{Z. Yang is with the Department of Engineering, King's College London, WC2R 2LS, UK, Email: \protect{yang.zhaohui@kcl.ac.uk}.}
\thanks{W. Saad is with the Wireless@VT, Bradley Department of Electrical and Computer Engineering, Virginia Tech, Blacksburg, VA, 24060, USA, Email: \protect{walids@vt.edu}.}
\thanks{C. Yin is with the Beijing Key Laboratory of Network System Architecture and Convergence,
Beijing University of Posts and Telecommunications, Beijing, 100876, China, Email: \protect{ccyin@ieee.org}.}
\thanks{S. Cui is with the Shenzhen Research Institute of Big Data and Future Network of Intelligence Institute, the Chinese University of Hong Kong, Shenzhen, 518172, China, Email: \protect{shuguangcui@cuhk.edu.cn}} 
\thanks{{The work was in part by the U.S. National Science Foundation under Grant CCF-1908308, and the U.S. Office of Naval Research under Grant N00014-15-1-2709.}}
 }
\maketitle
%
\begin{abstract}
In this paper, the problem of training federated learning (FL) algorithms over a realistic wireless network is studied. In the considered model, wireless users execute an FL algorithm while training their local FL models using their own data and transmitting the trained local FL models to a base station (BS) that generates a global FL model and sends the model back to the users. Since all training parameters are transmitted over wireless links, the quality of training is affected by wireless factors such as packet errors and the availability of wireless resources. 
 Meanwhile, due to the limited wireless bandwidth, the BS needs to select an appropriate subset of users to execute the FL algorithm so as to build a global FL model accurately.    
This joint learning, wireless resource allocation, and user selection problem is formulated as an optimization problem whose goal is to minimize an FL loss function that captures the performance of the FL algorithm.
 To seek the solution, a closed-form expression for the expected convergence rate of the FL algorithm is first derived to quantify the impact of wireless factors on FL. 
 Then, based on the expected convergence rate of the FL algorithm, the optimal transmit power for each user is derived, under a given user selection and uplink resource block (RB) allocation scheme. Finally, the user selection and uplink RB allocation is optimized so as to minimize the FL loss function.
Simulation results show that the proposed joint federated learning and communication framework
 can improve the identification accuracy by up to {\color{black}$1.4\%$, $3.5\%$ and $4.1\%$}, respectively, compared to: 1) An optimal user selection algorithm with random resource allocation, 2) a standard FL algorithm with random user selection and resource allocation, and {\color{black}3) a wireless optimization algorithm that minimizes the sum packet error rates of all users while being agnostic to the FL parameters.}
\end{abstract}

{\renewcommand{\thefootnote}{\fnsymbol{footnote}}
\footnotetext{A preliminary version of this work \cite{chen2019FLwireless} appears in the Proceedings of the 2019 IEEE Global Communications Conference.}}
%

\section{Introduction}
Standard machine learning approaches require centralizing the training data in a data center or a cloud \cite{chen2017machine,sun2018application,liu2019machine}. However, due to privacy concerns and limited communication resources for data transmission, it is undesirable for all users engaged in learning to transmit all of their collected data to a data center or a cloud. This, in turn, motivates the development of distributed learning frameworks that allow devices to use individually collected data to train a learning model locally.  
One of the most promising of such distributed learning frameworks is federated learning (FL) developed in \cite{bonawitz2019towards}. FL is a distributed machine learning method that enables users to collaboratively learn a shared prediction model while keeping their collected data on their devices \cite{NIPS2017_7029,8770530,saad2019vision,jeong2019multi,yang2019energy}. However, to train an FL algorithm over a wireless network, the users must transmit the training parameters over wireless links which can introduce training errors, due to the limited wireless resources (e.g., bandwidth) and the inherent unreliability of wireless links.
 \subsection{Related Works}
   

Recently, a number of existing works such as \cite{bonawitz2019towards} and\cite{li2019federated,konevcny2016federated,mcmahan2016communication,chen2018federated,konevcny2015federated,samarakoon2018distributed,ha2019coded,habachi2019fast,park2018wireless,zeng2019energy,8664630} have studied important problems related to the implementation of FL over wireless networks. The works in \cite{bonawitz2019towards} and \cite{li2019federated} provided a comprehensive survey on the design of FL algorithms and introduced various challenges, problems, and solutions for enhancing FL effectiveness. In \cite{konevcny2016federated}, the authors developed two update methods to reduce the uplink communication costs for FL. The work in \cite{mcmahan2016communication} presented a practical update method for a deep FL algorithm and conducted an extensive empirical evaluation for five different FL models using four datasets. An echo state network-based FL algorithm is developed in \cite{chen2018federated} to analyze and predict the location and orientation for wireless virtual reality users. In \cite{konevcny2015federated}, the authors proposed a novel FL algorithm that can minimize the communication cost. The authors in \cite{samarakoon2018distributed} studied the problem of joint power and resource allocation for ultra-reliable low latency communication in vehicular networks. The work in \cite{ha2019coded} developed a new approach to minimize the computing and transmission delay for FL algorithms.
In \cite{habachi2019fast}, the authors used FL algorithms for traffic estimation so as to maximize the data rates of users. 
While interesting, these prior works \cite{bonawitz2019towards} and \cite{li2019federated,konevcny2016federated,mcmahan2016communication,chen2018federated,konevcny2015federated,samarakoon2018distributed,ha2019coded,habachi2019fast} assumed that wireless networks can readily integrate FL algorithms. However, in practice, due to the unreliability of the wireless channels and to the wireless resource limitations (e.g., in terms of bandwidth and power), FL algorithms will encounter training errors due to the limitations of the wireless medium \cite{park2018wireless}. For example, symbol errors introduced by the unreliable nature of the wireless channel and by resource limitations can impact the quality and correctness of the FL updates among users. Such errors will, in turn, affect the performance of FL algorithms, as well as their convergence speed. Moreover, due to the wireless bandwidth limitations, the number of users that can perform FL is limited; a design issue that is not considered in \cite{bonawitz2019towards} and \cite{li2019federated,konevcny2016federated,mcmahan2016communication,chen2018federated,konevcny2015federated,samarakoon2018distributed,ha2019coded,habachi2019fast}. Furthermore, due to limited energy consumption of each user's device and strict delay requirement of FL, not all wireless users can participate in FL. Therefore, one must select {the appropriate subset of users to perform FL algorithms and optimize the performance of FL.}
 In practice, to effectively deploy FL over real-world wireless networks, it is necessary to investigate how the wireless factors affect the performance of FL algorithms. Here, we note that, although some works such as \cite{8770530} and {\cite{9085259,park2018wireless, zeng2019energy,8664630,vu2019cell,8737464,chen2020convergence,8851249}} have studied communication aspects of FL, these works are limited in several ways. First, the works in \cite{8770530}, \cite{park2018wireless}, and \cite{9085259} only provided a high-level exposition of the challenges of communication in FL. Meanwhile, the authors in {\cite{9085259,zeng2019energy,8664630,vu2019cell,8737464,chen2020convergence}} did not consider the effect of packet transmission errors on the performance of FL. The authors in \cite{8851249} developed an analytical model to characterize the effect of packet transmission errors on the FL performance. However, the work in \cite{8851249} focused attention on three specific scheduling policies and, hence, did not address optimal user selection and resource allocation to optimize the FL performance. 
 
\subsection{Contributions}
The main contribution of this paper is, thus, a novel framework for enabling the implementation of FL algorithms over wireless networks by jointly taking into account FL and wireless metrics and factors. To our best knowledge,   \emph{this is the first work that provides a comprehensive study of the connection between the performance of FL algorithms and the underlying wireless network.} Our key contributions include:

\begin{itemize}
\item We propose a novel FL model in which {cellular-connected} wireless users transmit their locally trained FL models to a base station (BS) that generates a global FL model and transmits it back to the users. For the considered FL model, the bandwidth for uplink transmission is limited and, hence, the BS needs to select appropriate users to execute the FL algorithm so as to minimize the FL loss function. In addition, the impact of the wireless packet transmission errors on the parameter update process of the FL model is explicitly considered.


\item 
In the developed joint communication and FL model, the BS must optimize its resource allocation and the users must optimize their transmit power allocation so as to decrease the packet error rates of each user thus improving the FL performance.
  To this end, we formulate this joint resource allocation and user selection problem for FL as an optimization problem whose goal is to minimize the training loss while meeting the delay and energy consumption requirements. Hence, our framework \emph{jointly considers learning and wireless networking metrics}.
  
\item To solve this problem, we first derive a closed-form expression for the expected convergence rate of the FL algorithm so as to build an explicit relationship between the packet error rates and the performance of the FL algorithm.
Based on this relationship, the optimization problem can be simplified as a mixed-integer nonlinear programming problem. To solve this simplified problem, we first find the optimal transmit power under given user selection and resource block (RB) allocation. Then, we transform the original optimization problem into a bipartite matching problem that is solved using a Hungarian algorithm which finds the optimal, FL-aware user selection and RB allocation strategy.


\item To further reduce the effect of packet transmission errors on the performance and convergence speed of FL, we perform fundamental analysis on an expression for the expected convergence rate of FL algorithms, which shows that the transmit power, RB allocation, and user selection will significantly affect the convergence speed and performance of FL algorithms. Meanwhile, by appropriately setting the learning rate and selecting the number of users that participate in FL, the effects of the transmission errors on FL algorithms can be reduced and the convergence of FL can be guaranteed.

\end{itemize}
Simulation results show that the transmit power, RB allocation, and the number of users will jointly affect the performance of FL over wireless networks. In particular, the simulation result shows that the proposed FL algorithm that considers the wireless factors can achieve up to 1.4\%, 3.5\%, and 4.1\% improvement in identification accuracy compared, respectively, to an optimal user selection algorithm with random resource allocation, a standard FL algorithm (e.g., such as in \cite{konevcny2016federated},) FL with random user selection and resource allocation, and a wireless optimization algorithm that minimizes the sum packet error rates of all users while being agnostic to the FL parameters. 

The rest of this paper is organized as follows. The system model and problem formulation are described in Section \uppercase\expandafter{\romannumeral2}. The expected convergence rate of FL algorithms is studied in Section \uppercase\expandafter{\romannumeral3}. The optimal resource allocation and user selection are determined in Section \uppercase\expandafter{\romannumeral4}. Simulation results are analyzed in Section \uppercase\expandafter{\romannumeral5}. Conclusions are drawn in Section \uppercase\expandafter{\romannumeral6}.

%
%
%
%

\section{System Model and Problem Formulation}\label{se:system}

%
%
\begin{figure}[!t]
  \begin{center}
   \vspace{0cm}
    \includegraphics[width=10cm]{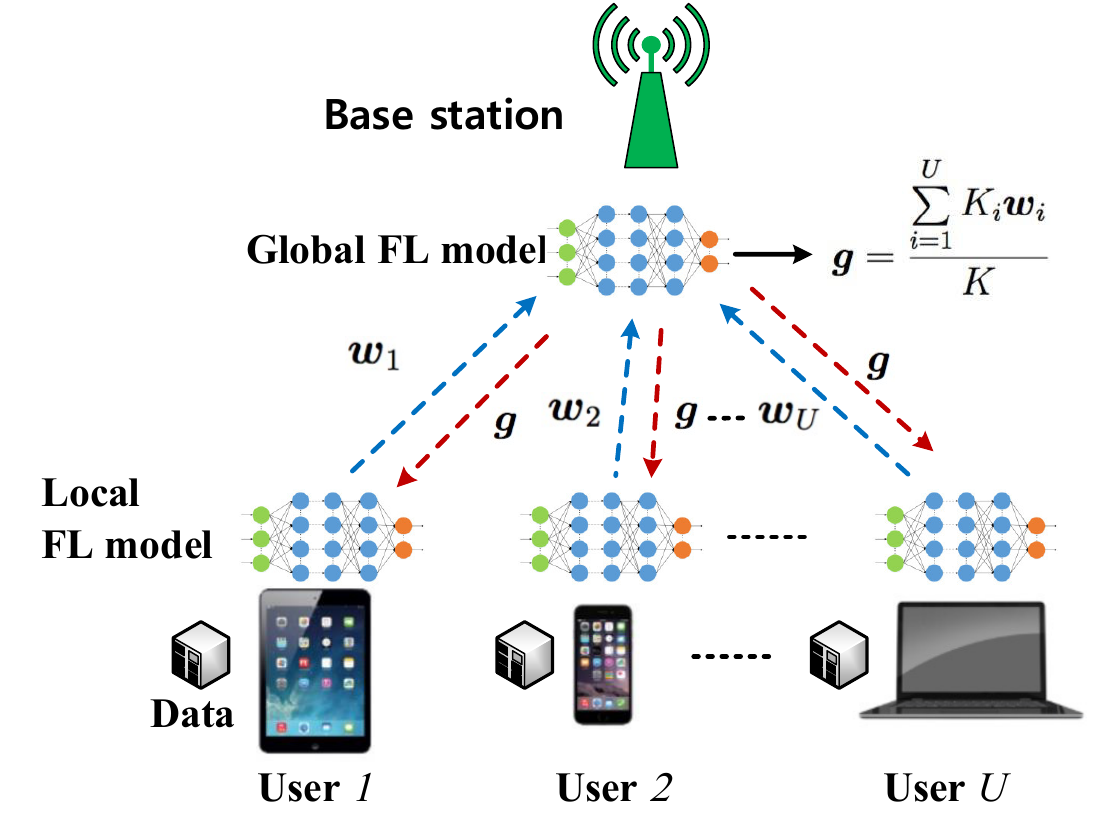}
    \vspace{-0.2cm}
    \caption{\label{model1} The architecture of an FL algorithm that is being executed over a wireless network with multiple devices and a single base station.}
  \end{center}\vspace{-0.2cm}
\end{figure}

\begin{figure}[!t]
  \begin{center}
   \vspace{0cm}
    \includegraphics[width=6cm]{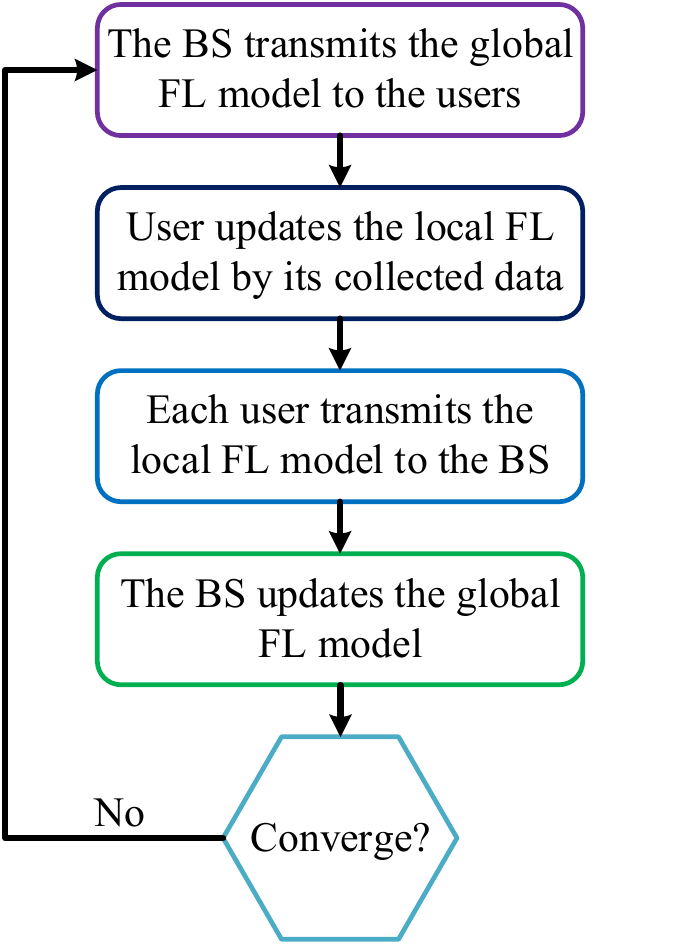}
    \vspace{-0.2cm}
    \caption{\label{FLprocedure} The learning procedure of an FL algorithm.}
  \end{center}\vspace{-0.7cm}
\end{figure}

Consider a cellular network in which one BS and a set $\mathcal{U}$ of $U$ users cooperatively perform an FL algorithm for data analysis and inference. For example, the network can execute an FL algorithm to sense the wireless environment and generate a holistic radio environment mapping \cite{8648450}. The use of FL for such applications is important because the data related to the wireless environment is distributed across the network \cite{saad2019vision} and the BS cannot collect all of this scattered data to implement a centralized learning algorithm. FL enables the BS and the users to collaboratively learn a shared learning model while keeping all of the training data at the device of each user.
In an FL algorithm, each user will use its collected training data to train an FL model. For example, for radio environment mapping, each user will collect the data related to the wireless environment for training an FL model. Hereinafter, the FL model that is trained at the device of each user (using the data collected by the user itself) is called the \emph{local FL model}. The BS is used to integrate the local FL models and generate a shared FL model. This shared FL model is used to improve the local FL model of each user so as to enable the users to collaboratively perform a learning task without training data transfer. Hereinafter, the FL model that is generated by the BS using the local FL models of its associated users is called the \emph{global FL model}.  As shown in Fig. \ref{model1}, the \emph{uplink} from the users to the BS is used to transmit the local FL model parameters while the \emph{downlink} is used to transmit the global FL model parameters.

\begin{table*}\footnotesize
  \newcommand{\tabincell}[2]{\begin{tabular}{@{}#1@{}}#2.8\end{tabular}}
\renewcommand\arraystretch{1.2}
 \caption{
    \vspace*{-0.1em}List of notations.}\vspace*{-0.6em}
\centering  
\begin{tabular}{|c||c|c||c|}
\hline
\textbf{Notation} & \textbf{Description} & \textbf{Notation} & \textbf{Description}\\
\hline
$U$ & Number of users & ${l}_{i}^\textrm{U}\left( \boldsymbol{r}_i,P_i \right)$& Uplink transmission delay\\
\hline
$\left( \boldsymbol{X}_i, \boldsymbol{y}_i\right)$ & Data collected by user $i$ &$\left(\boldsymbol{x}_{ik},{y}_{ik}\right)$ & Training data sample $k$ of user $i$  \\
\hline
$K$ & Total number of training data samples  &  $P_{\max}$ & Maximum transmit power of each user \\
\hline
 $P_B$ & Transmit power of the BS&${c}_{i}^\textrm{U}\left( \boldsymbol{r}_i,P_i \right)$ & Uplink data rate of user $i$ \\
\hline
$P_{i}$ & Transmit power of user $i$&${{K}_i}$ & Number of samples collected by user $i$  \\
\hline
$R$ & Number of RBs & $B^\textrm{D}$& Total downlink bandwidth of the BS\\
\hline
$\boldsymbol{g}$ & Global FL model& ${c}_{i}^\textrm{D}$ & Downlink data rate of user $i$ \\
\hline
$\mathcal{U}$ & Set of users & ${l}_{i}^\textrm{D}$ & Downlink transmission delay \\
\hline
{\color{black}$\boldsymbol{a}  \in \mathbb{R}^{1\times U}$ }& User selection vector& $Z\left( \boldsymbol{g}\right)$ & Data size of global FL model\\
\hline
  $\lambda$ & Learning rate & $q_i\left( \boldsymbol{r}_i,P_i   \right)$ & Packet error rate of user $i$\\
\hline
{\color{black}$\boldsymbol{R}\in \mathbb{R}^{R\times U}$}& RB allocation matrix of all users &$Z\left( \boldsymbol{w}_i\right)$ & Data size of local FL model \\
\hline
$\gamma_\textrm{T}$& Delay requirement  & $f\left(\boldsymbol{ g}\left(\boldsymbol{a}, \boldsymbol{R} \right), \boldsymbol{x}_{ik}, y_{ik}\right)$ & Loss function of FL\\
\hline
 $\boldsymbol{w}_i$ & Local FL model of user $i$ & $e_i\left(\boldsymbol{r}_i,P_i \right) $ & Energy consumption of user $i$ \\
\hline
$\gamma_\textrm{E}$ & Energy consumption requirement & {\color{black}$\boldsymbol{r}_{i}\in \mathbb{R}^{R\times 1}$} & RB allocation vector of user $i$ \\
\hline
$I_n$& Interference over RB $n$ & $B^\textrm{U}$ & Bandwidth of each RB \\
\hline
\end{tabular}
\vspace{-0.3cm}
\end{table*}  
 
 
\subsection{Machine Learning Model}
In our model, each user $i$ collects a marix $\boldsymbol{X}_i=\left[\boldsymbol{x}_{i1}, \ldots, \boldsymbol{x}_{iK_i}\right]$ of input data, where $K_i$ is the number of the samples collected by each user $i$ and each element $\boldsymbol{x}_{ik}$ is an input vector of the FL algorithm. The size of $\boldsymbol{x}_{ik}$ depends on the specific FL task. Our approach, however, is applicable to any generic FL algorithm and task.
Let ${y}_{ik}$ be the output of $\boldsymbol{x}_{ik}$. For simplicity, we consider an FL algorithm with a single output, however, our approach can be readily generalized to a case with multiple outputs \cite{konevcny2016federated}.
 The output data vector for training the FL algorithm of user $i$ is $\boldsymbol{y}_i=\left[{y}_{i1}, \ldots, {y}_{iK_i}\right]$.
We define a vector $\boldsymbol{w}_i$ to capture the parameters related to the local FL model that is trained by $\boldsymbol{X}_i$ and $\boldsymbol{y}_i$. In particular, $\boldsymbol{w}_i$ determines the local FL model of each user $i$. For example, in a linear regression learning algorithm, $\boldsymbol{x}_{ik}^{\rm{T}}{\boldsymbol{w}_i}$ represents the predicted output and $\boldsymbol{w}_i$ is a weight vector that determines the performance of the linear regression learning algorithm. {\color{black}For each user $i$, the local training problem seeks to find the optimal learning model parameters $\boldsymbol{w}_i^*$ that minimize its training loss. 
 The training process of an FL algorithm is done in a way to solve the following optimization problem:
\addtocounter{equation}{0}
\begin{equation}\label{eq:ML}
\begin{split}
\mathop {\min } \limits_{{\boldsymbol{w}_1, \ldots, \boldsymbol{w}_U}} \frac{1}{K}  \sum\limits_{i = 1}^U  \sum\limits_{k=1 }^{K_i} {f\left( {{\boldsymbol{w}_i},{\boldsymbol{x}_{ik}},{{y}_{ik}}} \right)} ,
\end{split}
\end{equation}
\vspace{-0.3cm}
\begin{align}\label{c1}
&\!\!\!\!\!\!\!\!\rm{s.\;t.}\;\scalebox{1}{$ \boldsymbol{w}_1=\boldsymbol{w}_2=\ldots=\boldsymbol{w}_U=\boldsymbol{g},$}\tag{\theequation a}
\end{align}
where $K=\sum\limits_{i = 1}^U K_i$ is total size of training data of all users and $\boldsymbol{g}$ is the global FL model that is generated by the BS and ${f\left( {{\boldsymbol{w}_i},{\boldsymbol{x}_{ik}},{{y}_{ik}}} \right)}$ is a loss function. The loss function captures the performance of the FL algorithm. For different learning tasks, the FL performance captured by the loss function is different. For example, for a prediction learning task, the loss function captures the prediction accuracy of FL. In contrast, for a classification learning task, the loss function captures the classification accuracy. Meanwhile, for different FL algorithms, different loss functions can be defined \cite{hennig2007some}. For example, for a linear regression FL, the loss function is ${f\left( {{\boldsymbol{w}_i},{\boldsymbol{x}_{ik}},{{y}_{ik}}} \right)}=\frac{1}{2}\left( \boldsymbol{x}_{ik}^{\rm{T}}{\boldsymbol{w}_i} -{y}_{ik} \right)^2$. As the prediction errors (i.e., $\boldsymbol{x}_{ik}^{\rm{T}}{\boldsymbol{w}_i} -{y}_{ik}$) increase, the loss function ${f\left( {{\boldsymbol{w}_i},{\boldsymbol{x}_{ik}},{{y}_{ik}}} \right)}$ increases. 
Constraint (\ref{eq:ML}a) is used to ensure that, once the FL algorithm converges, all of the users and the BS will share the same FL model for their learning task. This captures the fact that the purpose of an FL algorithm is to enable the users and the BS to learn an optimal global FL model without data transfer.   
To solve (\ref{eq:ML}), the BS will transmit the parameters $\boldsymbol{g}$ of the global FL model to its users so that they train their local FL models. Then, the users will transmit their local FL models to the BS to update the global FL model.
The detailed procedure of training an FL algorithm \cite{45648} to minimize the loss function in (\ref{eq:ML}) is shown in Fig. \ref{FLprocedure}. 
In FL, the update of each user $i$'s local FL model $\boldsymbol{w}_i$ depends on the global model $\boldsymbol{g}$ while the update of the global model $\boldsymbol{g}$ depends on all of the users' local FL models. The update of the  local FL model $\boldsymbol{w}_i$ depends on the learning algorithm. For example, one can use gradient descent, stochastic gradient descent, or randomized coordinate descent \cite{konevcny2016federated} to update the local FL model. The update of the global model $\boldsymbol{g}$ is given by \cite{konevcny2016federated}
\begin{equation}\label{eq:w}
\boldsymbol{g}_{t}=\frac{\sum\limits_{i = 1}^U K_i \boldsymbol{w}_{i,t}}{K}.
\end{equation}   
During the training process, each user will first use its training data $\boldsymbol{X}_i$ and $\boldsymbol{y}_i$ to train the local FL model $\boldsymbol{w}_i$ and then, it will transmit $\boldsymbol{w}_i$ to the BS via wireless cellular links. Once the BS receives the local FL models from all participating users, it will update the global FL model based on (\ref{eq:w}) and transmit the global FL model $\boldsymbol{g}$ to all users to optimize the local FL models. As time elapses, the BS and users can find their optimal FL models and use them to minimize the loss function in (\ref{eq:ML}).
Since all of the local FL models are transmitted over wireless cellular links, once they are received by the BS, they may contain erroneous symbols due to the unreliable nature of the wireless channel, which, in turn, will have a significant impact on the performance of FL. Meanwhile, the BS must update the global FL model once it receives all of the local FL models from its users and, hence, the wireless transmission delay will significantly affect the convergence of the FL algorithm. In consequence, to deploy FL over a wireless network, \emph{one must jointly consider the wireless and learning performance and factors.}


\subsection{Transmission Model}  
For uplink, we assume that an orthogonal frequency division multiple access (OFDMA) technique in which each user occupies one RB.
The uplink rate of user $i$ transmitting its local FL model parameters to the BS is given by
\begin{equation}\label{eq:uplinkdatarate}
c_{i}^\textrm{U} \left(\boldsymbol{r}_{i}, P_{i}\right)=  \sum\limits_{n = 1}^R r_{i,n} B^\textrm{U}\mathbb E_{h_i}\left( {\log _2}\left(\!1\!+\! {\frac{{{P_{i}}{h_{i}}}}{I_n+B^\textrm{U}N_0}} \!\right)\right),
\end{equation}
 where $\boldsymbol{r}_{i}=\left[{r}_{i,1},\ldots, {r}_{i,R}\right]$ is an RB allocation vector with $R$ being the total number of RBs, $r_{i,n} \in \left\{0,1\right\}$ and $ \sum\limits_{n = 1}^R r_{i,n}=1$; $r_{i,n}=1$ indicates that RB $n$ is allocated to user $i$, and $r_{i,n}=0$, otherwise; $B^\textrm{U}$ is the bandwidth of each RB and 
 $P_i$ is the transmit power of user $i$; 
{\color{black}$ h_{i}=o_{i}d_{i}^{-2}$ is the channel gain between user $i$ and the BS with $d_i$ being the distance between user $i$ and the BS and $o_{i}$ being the Rayleigh fading parameter}; {\color{black} $\mathbb E_{h_i} \left( \cdot \right)$ is the expectation with respect to $h_i$}; 
$N_0$ is the noise power spectral density; $I_n$ is the interference caused by the users that are located in other service areas (e.g., other BSs not participating in the FL algorithm) and use RB $n$. Note that, although we ignore the optimization of resource allocation for the users located at the other service areas, we must consider the interference caused by the users in other service areas (if they are sharing RBs with the considered FL users), since this interference may significantly affect the packet error rates and the performance of FL.  
 
Similarly, the downlink data rate achieved by the BS when transmitting the parameters of global FL model to each user $i$ is given by
 \begin{equation}\label{eq:downlinkdatarate}
 {\color{black}
c_{i}^\textrm{D} =  B^\textrm{D}\mathbb E_{h_i}\left( {\log _2}\left(\!1\!+\! {\frac{{{P_{B}}{ h_{i}}}}{ I^\textrm{D}  +B^\textrm{D}N_0}} \!\right)\right),}
\end{equation}
where $B^\textrm{D}$ is the bandwidth that the BS used to broadcast the global FL model of each user $i$; $P_B$ is the transmit power of the BS; $I^\textrm{D}$ is the interference caused by other BSs not participating in the FL algorithm.
 Given the uplink data rate $c_i^\textrm{U}$ in (\ref{eq:uplinkdatarate}) and the downlink data rate $c_i^\textrm{D}$ in (\ref{eq:downlinkdatarate}), the transmission delays between user $i$ and the BS over uplink and downlink are respectively specified as
\begin{equation}
l_{i}^\textrm{U}\left(\boldsymbol{r}_{i}, P_{i}\right) =\frac{Z\left(\boldsymbol{w}_i\right) }{c_{i}^\textrm{U} \left(\boldsymbol{r}_{i}, P_{i}\right)    },
\end{equation}
\begin{equation}
l_{i}^\textrm{D}=\frac{Z\left(\boldsymbol{g} \right)}{c_{i}^\textrm{D}},
\end{equation}
where function $Z\left( \boldsymbol{x} \right)$ is the data size of $\boldsymbol{x}$ which is defined as the number of bits that the users or the BS require to transmit vector $\boldsymbol{x}$ over wireless links. In particular, 
$Z\left(\boldsymbol{w}_i\right)$ represents the number of bits that each user $i$ requires to transmit local FL model $\boldsymbol{w}_i$ to the BS while $Z\left(\boldsymbol{g}\right)$ is the number of bits that the BS requires to transmit the global FL model $\boldsymbol{g}$ to each user. Here, $Z\left(\boldsymbol{w}_i\right)$ and $Z\left(\boldsymbol{g}\right)$ are determined by the type of implemented FL algorithm. From (\ref{eq:w}), we see that the number of elements in the global FL model $\boldsymbol{g}$ is similar to that of each user $i$'s local FL model $\boldsymbol{w}_i$. Hence, we assume $Z\left(\boldsymbol{w}_i\right)=Z\left(\boldsymbol{g}\right)$.

\subsection{Packet Error Rates}
For simplicity, we assume that each local FL model $\boldsymbol{w}_i$ will be transmitted as a single packet in the uplink. A cyclic redundancy check (CRC) mechanism is used to check the data errors in the received local FL models at the BS. In particular, $C\left( \boldsymbol{w}_i \right)=0$ indicates that the local FL model received by the BS contains data errors; otherwise, we have $C\left( \boldsymbol{w}_i \right)=1$. The packet error rate experienced by the transmission of each local FL model $\boldsymbol{w}_i$ to the BS is given by \cite{5703199}
\begin{equation}\label{eq:per}
q_i\left(\boldsymbol{r}_{i}, P_{i}\right)= \sum\limits_{n = 1}^R r_{i,n}q_{i,n},
\end{equation}
where {\color{black}$q_{i,n}=\mathbb E_{h_i}\left(1-\exp\left(-\frac{m\left({I_n+B^\textrm{U}N_0}\right)}{{{{{P_{i}}{h_{i}}}}}} \right) \right)$} is the packet error rate over RB $n$ with $m$ being a waterfall threshold \cite{5703199}.

In the considered system, whenever the received local FL model contains errors, the BS will not use it for the update of the global FL model. We also assume that the BS will not ask the corresponding users to resend their local FL models when the received local FL models contain data errors. Instead, the BS will directly use the remaining correct local FL models to update the global FL model. As a result, the global FL model in (\ref{eq:w}) can be written as 
\begin{equation}\label{eq:globalww}
\boldsymbol{ g}\left(\boldsymbol{a}, \boldsymbol{P},\boldsymbol {R}\right)=\frac{\sum\limits_{i = 1}^U K_ia_i \boldsymbol{ w}_{i}C\left( \boldsymbol{w}_i \right)}{{\sum\limits_{i = 1}^U K_ia_iC\left( \boldsymbol{w}_i \right)}},
\end{equation}
where 
\begin{equation} \label{itproofeq2}
C\left( \boldsymbol{w}_i \right)=\left\{ \begin{array}{ll}
\!\!1, &\text{with probability}\; 1-q_i\left(\boldsymbol{r}_{i}, P_{i}\right),\\
\!\!0, &\text{with probability}~q_i\left(\boldsymbol{r}_{i}, P_{i}\right),
\end{array} \right.
\end{equation}
$\boldsymbol{a}=[a_1,\ldots, a_U]$ is the vector of the user selection index with $a_i=1$ indicating that user $i$ performs the FL algorithm and $a_i=0$, otherwise, $\boldsymbol R=[\boldsymbol r_1, \cdots, \boldsymbol r_U]$, $\boldsymbol P=[ P_1, \cdots,  P_U]$, $\sum\limits_{i = 1}^U K_ia_iC\left( \boldsymbol{w}_i \right)$ is the total number of training data samples, which depends on the user selection vector $\boldsymbol{a}$ and packet transmission $C\left( \boldsymbol{w}_i \right)$, $K_ia_i\boldsymbol{ w}_{i}C\left( \boldsymbol{w}_i \right)=0$ indicates that the local FL model of user $i$ contains data errors and, hence, the BS will not use it to generate the global FL model, and $\boldsymbol{g}\left(\boldsymbol{a},\boldsymbol{P}, \boldsymbol {R}\right)$ is the global FL model that explicitly incorporates the effect of wireless transmission. 
From (\ref{eq:globalww}), we see that the global FL model also depends on the resource allocation matrix $\boldsymbol {R}$, user selection vector $\boldsymbol{a}$, and transmit power vector $\boldsymbol{P}$.


%


\subsection{Energy Consumption Model}
In our network, the energy consumption of each user consists of the energy needed for two purposes: a) Transmission of the local FL model  and b) Training of the local FL model.
The energy consumption of each user $i$ is given by \cite{8057276}
\begin{equation}\label{eq:energy}
\begin{split}
e_{i}\left( \boldsymbol{r}_{i},P_{i}\right) =\varsigma \omega_i \vartheta^2 Z\left( \boldsymbol{w}_{i}\right)+P_{i}l_{i}^\textrm{U}\left( \boldsymbol{r}_{i}, P_{i}\right),
\end{split}
\end{equation} 
where $\vartheta$ is the frequency of the central processing unit (CPU) clock of each user $i$, $\omega_i$ is the number of CPU cycles required for computing per bit data of user $i$, which is assumed to be equal for all users, and
 $\varsigma$ is the energy consumption 
coefficient depending on the chip of each user $i$'s device \cite{8057276}. In (\ref{eq:energy}), $\varsigma \omega_i \vartheta^2  Z\left( \boldsymbol{w}_{i}\right)$ is the energy consumption of user $i$ training the local FL model at its own device and $P_{i}l_{i}^\textrm{U}\left( \boldsymbol{r}_{i}, P_{i}\right)$ represents the energy consumption of local FL model transmission from user $i$ to the BS. Note that, since the BS can have continuous power supply, we do not consider the energy consumption of the BS in our optimization problem.

\subsection{Problem Formulation}
%
To jointly design the wireless network and the FL algorithm, we now formulate an optimization problem whose goal is to minimize the training loss, while factoring in the wireless network parameters. This minimization problem includes optimizing transmit power allocation as well as resource allocation for each user.  The minimization problem is given by
\addtocounter{equation}{0}
\begin{equation}\label{eq:max}
\begin{split}
\mathop {\min }\limits_{\boldsymbol{a}, \boldsymbol{P}, \boldsymbol{R}}  \frac{1}{K} \sum\limits_{i=1}^{U}  \sum\limits_{k =1 }^{K_i}{f\left( {\boldsymbol{g}\left(\boldsymbol{a}, \boldsymbol{P}, \boldsymbol{R}\right),{\boldsymbol{x}_{ik}},{{y}_{ik}}} \right)},
\end{split}
\end{equation}
\vspace{-0.3cm}
\begin{align}\label{c1}
\setlength{\abovedisplayskip}{-20 pt}
\setlength{\belowdisplayskip}{-20 pt}
&\!\!\!\!\!\!\!\!\rm{s.\;t.}\;\scalebox{1}{$a_i, r_{i,n} \in \left\{0,1\right\}, \;\;\;\;\;\forall i \in \mathcal{U}, n=1,\ldots, R,$}\tag{\theequation a}\\
&\scalebox{1}{$\;\;\; \sum\limits_{n = 1}^R r_{i,n}=a_i,\;\;\forall i \in \mathcal{U}, $} \tag{\theequation b}\\
&\scalebox{1}{$\;\;\;   l_{i}^\textrm{U}\left(\boldsymbol{r}_{i}, P_{i}\right)+l_{i}^\textrm{D} \le \gamma_\textrm{T}  ,\;\;\forall i \in \mathcal{U}, $} \tag{\theequation c}\\
&\scalebox{1}{$\;\;\;  e_{i}\left(\boldsymbol{r}_{i},P_{i}\right)\le \gamma_\textrm{E} ,\;\;\forall i \in \mathcal{U},$} \tag{\theequation d}\\
&\scalebox{1}{$\;\;\; \sum\limits_{i \in \mathcal{U}} {r_{i,n}}  \le 1,~\forall n=1,\ldots, R, $} \tag{\theequation e}\\
&\scalebox{1}{$\;\;\;  0 \le P_{i} \le P_{\max},\;\;\;\;\;\forall i \in \mathcal{U},$}\tag{\theequation f}
\end{align}
where
$\gamma_\textrm{T}$ is the delay requirement for implementing the FL algorithm, $\gamma_\textrm{E}$ is the energy consumption of the FL algorithm, and $B$ is the total downlink bandwidth. (\ref{eq:max}a) and (\ref{eq:max}b) indicates that each user can occupy only one RB for uplink data transmission. (\ref{eq:max}c) is the delay needed to execute the FL algorithm at each learning step. (\ref{eq:max}d) is the energy consumption requirement of performing an FL algorithm at each learning step. 
 (\ref{eq:max}e) indicates that each uplink RB can be allocated to at most one user. (\ref{eq:max}f) is a maximum transmit power constraint. {\color{black} From (\ref{eq:max}), we can see that the user selection vector $\boldsymbol{a}$, the RB allocation matrix $\boldsymbol{R}$, and the transmit power vector $\boldsymbol{P}$ will not change during the FL training process and the optimized $\boldsymbol{a}$, $\boldsymbol{R}$, and $\boldsymbol{P}$ must meet the delay and energy consumption requirements at each learning step in (\ref{eq:max}c) and (\ref{eq:max}d).  }

From (\ref{eq:per}) and (\ref{eq:globalww}), we see that the transmit power and resource allocation determine the packet error rate, thus affecting the update of the global FL model. In consequence, the loss function of the FL algorithm in (\ref{eq:max}) depends on the resource allocation and transmit power. Moreover, (\ref{eq:max}c) shows that, in order to perform an FL algorithm, the users must satisfy a specific delay requirement. In particular, in an FL algorithm, the BS must wait to receive the local model of each user before updating its global FL model. Hence, transmission delay plays a key role in the FL performance. In a practical FL algorithm, it is desirable that all users transmit their local FL models to the BS simultaneously. From (\ref{eq:max}d), we see that to perform the FL algorithm, a given user must have enough energy to transmit and update the local FL model throughout the FL iterative process. If this given user does not have enough energy, the BS should choose this user to participate in the FL process.
In consequence, in order to implement an FL algorithm in a real-world network, the wireless network must provide low energy consumption and latency, and highly reliable data transmission.  

\section{Analysis of the FL Convergence Rate}
To solve (\ref{eq:max}), we first need to analyze how the packet error rate affects the performance of the FL. To find the relationship between the packet error rates and the FL performance, we must first analyze the convergence rate of FL. However, since the update of the global FL model depends on the instantaneous signal-to-interference-plus-noise ratio (SINR), we can analyze only the expected convergence rate of FL.  Here, we first analyze the expected convergence rate of FL. Then, we show how the packet error rate affects the performance of the FL in (\ref{eq:max}).

 In the studied network, the users adopt a standard gradient descent method to update their local FL models as done in \cite{konevcny2016federated}. {\color{black}Therefore, during the training process, the local FL model $\boldsymbol{w}_i$ of each {selected} user $i$ ($a_i=1$) at step $t$ is}
\begin{equation}\label{eq:update}
\boldsymbol{w}_{i,t+1}=\boldsymbol{ g}_t\left(\boldsymbol{a},\boldsymbol{P},  \boldsymbol{R} \right) -\frac{\lambda}{K_i} \sum\limits_{k =1 }^{K_i} \nabla {f\left( { \boldsymbol{ g}_t\left(\boldsymbol{a},\boldsymbol{P}, \boldsymbol{R}\right) ,{\boldsymbol{x}_{ik}},{{y}_{ik}}} \right)},
\end{equation}  
where $\lambda$ is the learning rate and $\nabla {f\left( { \boldsymbol{ g}_t\left(\boldsymbol{a},\boldsymbol{P}, \boldsymbol{R}\right) ,{\boldsymbol{x}_{ik}},{{y}_{ik}}} \right)}$ is the gradient of ${f\left( {\boldsymbol{ g}_t\left(\boldsymbol{a}, \boldsymbol{P}, \boldsymbol{R}\right) ,{\boldsymbol{x}_{ik}},{{y}_{ik}}} \right)}$ with respect to $\boldsymbol{ g}_t\left(\boldsymbol{a},\boldsymbol{P}, \boldsymbol{R}\right) $. 

We assume that $F\left(\boldsymbol{g}\right)=\frac{1}{K} \sum\limits_{i=1}^{U}  \sum\limits_{k =1 }^{K_i}{f\left( {\boldsymbol{g},{\boldsymbol{x}_{ik}},{{y}_{ik}}} \right)}$ and $F_i\left(\boldsymbol{g}\right)=\sum\limits_{k =1 }^{K_i}{f\left( {\boldsymbol{g},{\boldsymbol{x}_{ik}},{{y}_{ik}}} \right)}$ where $\boldsymbol{g}$ is short for
$\boldsymbol{g}\left(\boldsymbol{a}, \boldsymbol{P}, \boldsymbol{R}\right)$. Based on (\ref{eq:update}), the update of global FL model $\boldsymbol{g}$ at step $t$ is given by
\begin{equation} \label{itproofeq1}
\boldsymbol{g}_{t+1}= \boldsymbol{g}_{t} -\lambda \left(\nabla F \left(\boldsymbol{g}_{t} \right)-\boldsymbol{o}\right),
\end{equation} 
where {$\color{black}\boldsymbol{o}=\nabla F \left(\boldsymbol{g}_{t} \right)-\frac{\sum\limits_{i = 1}^U a_i \sum\limits_{k =1 }^{K_i}{\nabla f\left( {\boldsymbol{g},{\boldsymbol{x}_{ik}},{{y}_{ik}}} \right)} C\left( \boldsymbol{w}_i \right)}{{\sum\limits_{i = 1}^U K_ia_iC\left( \boldsymbol{w}_i \right)}}$}.
We also assume that the FL algorithm converges to an optimal global FL model $\boldsymbol{g}^*$ after the learning steps. To derive the expected convergence rate of FL, we first make the following assumptions, {\color{black}as done in \cite{8851249,yang2019energy}}.
\begin{itemize}
\item First, we assume that the gradient $\nabla F\left( \boldsymbol{g}\right)$ of $F\left( \boldsymbol{g}\right)$ is uniformly Lipschitz continuous with respect to $\boldsymbol{g}$ \cite{friedlander2012hybrid}. Hence, we have
\begin{equation}\label{itproofas1}
\|\nabla  F \left(\boldsymbol{g}_{t+1} \right)  - \nabla F\left( \boldsymbol{g}_{t} \right)\|
\leq  L\| \boldsymbol{g}_{t+1} - \boldsymbol{g}_{t}\|,
\end{equation}
where $L$ is a positive constant and $\| \boldsymbol{g}_{t+1} - \boldsymbol{g}_{t}\|$ is the norm of $ \boldsymbol{g}_{t+1} - \boldsymbol{g}_{t}$.    

\item Second, we assume that $F\left( \boldsymbol{g}\right)$ is strongly convex with positive parameter $\mu$, such that
\begin{equation}\label{itproofas2}
F \left(\boldsymbol{g}_{t+1} \right) \geq F \left(\boldsymbol{g}_{t} \right)
+\left( \boldsymbol{g}_{t+1} - \boldsymbol{g}_{t}\right)^{T}  \nabla F \left(\boldsymbol{g}_{t} \right)
+\frac  {\mu } 2 \| \boldsymbol{g}_{t+1} - \boldsymbol{g}_{t}\|^2.
\end{equation}

\item We also assumed that $F\left( \boldsymbol{g}\right)$ is twice-continuously differentiable. 
Based on \eqref{itproofas1} and \eqref{itproofas2}, we have
\begin{equation}\label{itproofas2_2}
 \mu \boldsymbol{I} \preceq
\nabla^2 F \left(\boldsymbol{g} \right) \preceq L \boldsymbol{I}.
\end{equation}

\item We also assume that $\|\nabla {f\left( { \boldsymbol{g}_{t} ,{\boldsymbol{x}_{ik}},{{y}_{ik}}} \right)}\|^2\leq \zeta_1+\zeta_2 \| \nabla F \left(\boldsymbol{g}_{t} \right)\|^2$ with $\zeta_1, \zeta_2\ge 0$.

\end{itemize} 
{\color{black}These assumptions can be satisfied by several widely used loss functions such as the mean squared error, logistic regression, and cross entropy \cite{friedlander2012hybrid}. These popular loss functions can be used to capture the performance of implementing practical FL algorithms for identification, prediction, and classification. For future work, we can investigate how to extend our work for other non-convex loss functions.} The expected convergence rate of the FL algorithms can now be obtained by the following theorem. 

\begin{theorem}\label{th:1}
\emph{Given the transmit power vector $\boldsymbol{P}$, RB allocation matrix $\boldsymbol{R}$, user selection vector $\boldsymbol{a}$, optimal global FL model $\boldsymbol{g}^*$, and the learning rate $\lambda=\frac{1}{L}$, the upper bound of $\mathbb E \left(F\left(\boldsymbol{g}_{t+1}\right) -  F\left(\boldsymbol{g}^*\right)\right)$ can be given by
\begin{equation}\label{eq:theorem1}
{\color{black}
\begin{split}
&\mathbb E \left(F\left(\boldsymbol{g}_{t+1}\right)-F\left(\boldsymbol{g}^{*}\right)\right)  \leq A^t  \mathbb E\left(F \left(\boldsymbol{g}_{0} \right)-F\left(\boldsymbol{g}^{*}\right)\right)
+\underbrace{\frac  {2\zeta_1}  {LK}\sum\limits_{i = 1}^U K_i  \left(1-a_i+a_iq_i\left(\boldsymbol{r}_{i}, P_{i} \right)\right) \frac{1-A^t}{1-A}}_{\textrm{Impact of wireless factors on FL convergence}},
 \end{split}
 }
\end{equation}
where {\color{black}$A=1-\frac{\mu}{L}+ \frac  {4\mu\zeta_2} { {LK}}\sum\limits_{i = 1}^U K_i \left(1-a_i+a_iq_i\left(\boldsymbol{r}_{i}, P_{i} \right)\right)$} and {\color{black} $\mathbb E \left( \cdot \right)$ is the expectation with respect to packet error rate}.
}
\end{theorem}
\begin{proof} See Appendix A.
\end{proof}
{\color{black}In Theorem \ref{th:1}, $\boldsymbol{g}_{t+1}$ is the global FL model that is generated based only on the the local FL models of selected users ($a_i=1$) at step $t+1$. $\boldsymbol{g}^{*}$ is the optimal FL model that is generated based on the local FL models of all uses in an ideal setting with no wireless errors. From Theorem \ref{th:1}, we see that a gap, $\frac  {2\zeta_1}  {LK}\sum\limits_{i = 1}^U K_i  \left(1-a_i+a_iq_i\left(\boldsymbol{r}_{i}, P_{i} \right)\right)\frac{1-A^t}{1-A}$, exists between $\mathbb E \left( F\left(\boldsymbol{g}_{t}\right)\right)  $ and $\mathbb E \left(F\left(\boldsymbol{g}^{*}\right)\right)$. This gap is caused by the packet errors and the user selection policy.} As the packet error rate decreases, the gap between $\mathbb E \left(F\left(\boldsymbol{g}_{t}\right)\right)  $ and $\mathbb E \left( F\left(\boldsymbol{g}^{*}\right)\right)$ decreases. Meanwhile, as the number of users that implement the FL algorithm increases, the gap also decreases. Moreover, as the packet error rate decreases, the value of $A$ also decreases, which indicates that the convergence speed of the FL algorithm improves. Hence, it is necessary to optimize resource allocation, user selection, and transmit power for the implementation of any FL algorithm over a realistic wireless network. {\color{black}Theorem \ref{th:1} can be extended to the case in which each local FL model needs to be transmitted over a large number of packets by replacing the packet error rate $q_i\left(\boldsymbol{r}_{i}, P_{i} \right)$ in (\ref{eq:theorem1}) with the error rate of transmitting multiple packets to send the entire local FL model.}

According to Theorem \ref{th:1}, the following result is derived to guarantee the convergence of the FL algorithm. 
 \begin{proposition}\label{pr:1}
\emph{Given the learning rate $\lambda=\frac{1}{L}$, to guarantee convergence and simplify the optimization problem in (\ref{eq:max}), $\zeta_2$ must satisfy
\begin{equation}\label{eq:pr1}
{\color{black}
0<\zeta_2 < \frac{K}{\mathop {\max }\limits_{\boldsymbol{P}, \boldsymbol{R}} 4\sum\limits_{i=1}^U  K_i q_i\left(\boldsymbol{r}_{i}, P_{i} \right)}.}
\end{equation}
}
\end{proposition}
\begin{proof}
From Theorem \ref{th:1}, we see that when $A<1$, $A^t=0$. Hence, $\mathbb E \left( F\left(\boldsymbol{g}_{t+1}\right)-F\left(\boldsymbol{g}^{*}\right)\right) = \sum\limits_{i=1}^U K_i q_i\left(\boldsymbol{r}_{i}, P_{i} \right) \frac{1}{1-A}$ and the FL algorithm converges. In consequence, to guarantee the convergence, we only need to make $A=1-\frac{\mu}{L}+ \frac  {4\mu\zeta_2} { {LK}}\sum\limits_{i=1}^U  K_i q_i\left(\boldsymbol{r}_{i}, P_{i} \right)<1$. 
From (\ref{itproofas2_2}), we see that $\mu<L$ and, hence, $\frac{\mu}{L}<1$. To make $A<1$, we only need to ensure that $\frac{4\mu\zeta_2}{LK}\sum\limits_{i=1}^U  K_i q_i\left(\boldsymbol{r}_{i}, P_{i} \right)-\frac{\mu}{L}<0$. Therefore, we have $\zeta_2<\frac{K}{4\sum\limits_{i=1}^U  K_i q_i\left(\boldsymbol{r}_{i}, P_{i} \right)}$. {\color{black}To simplify the optimization problem in (\ref{eq:max}), $\zeta_2$ must satisfy $\zeta_2<\frac{K}{4\sum\limits_{i=1}^U  K_i q_i\left(\boldsymbol{r}_{i}, P_{i} \right)}$ for all RB allocation schemes. Hence, we choose this parameter such that $\zeta_2 < \frac{K}{\mathop {\max }\limits_{\boldsymbol{P}, \boldsymbol{R}} 4\sum\limits_{i=1}^U  K_i q_i\left(\boldsymbol{r}_{i}, P_{i} \right)}$}. Since $\zeta_2$ must satisfy $\|\nabla {f\left( { \boldsymbol{g}_{t} ,{\boldsymbol{x}_{ik}},{{y}_{ik}}} \right)}\|^2\leq \zeta_1+\zeta_2 \| \nabla F \left(\boldsymbol{g}_{t} \right)\|^2$, we have $\zeta_2>0$.    This completes the proof.
\end{proof}

From Proposition \ref{pr:1}, we see that the convergence of the FL algorithm depends on the parameters related to the approximation of $ \| \nabla F \left(\boldsymbol{g}_{t} \right)\|^2$. Using Proposition \ref{pr:1}, we can determine the convergence of the FL algorithm based on the approximation of $ \| \nabla F \left(\boldsymbol{g}_{t} \right)\|^2$. 

Based on Theorem \ref{th:1}, next, we can also derive the convergence rate of an FL algorithm when there are no packet errors. 

\begin{lemma}\label{le:1}
\emph{Given the optimal global FL model $\boldsymbol{g}^*$ and the learning rate $\lambda=\frac{1}{L}$, the upper bound of $\mathbb E  \left(F\left(\boldsymbol{g}_{t+1}\right) -  F\left(\boldsymbol{g}^*\right)\right)$ for an FL algorithm without considering packet errors and user selection is given by
\begin{equation}\label{eq:lemma1}
\begin{split}
\mathbb E \left( F\left(\boldsymbol{g}_{t+1}\right)-F\left(\boldsymbol{g}^{*}\right)\right)  \leq \left(1-\frac{\mu}{L} \right)^t  \mathbb E\left(F \left(\boldsymbol{g}_{0} \right)-F\left(\boldsymbol{g}^{*}\right)\right).
 \end{split}
\end{equation}
}
\end{lemma}
\begin{proof}
Since the FL algorithms do not consider the packet error rates and user selection, we have $q_i\left(\boldsymbol{r}_{i}, P_{i}\right)=0$, $a_i=1$,  $A=1-\frac{\mu}{L}$. Hence, $ \frac  {2\zeta_1}  {LK}\sum\limits_{i = 1}^U K_i  \left(1-a_i+a_iq_i\left(\boldsymbol{r}_{i}, P_{i} \right)\right) \frac{1-A^t}{1-A}=0$. Then (\ref{eq:lemma1}) can be derived based on (\ref{eq:theorem1}). 
\end{proof}
From Lemma \ref{le:1}, we can observe that, if we do not consider the packet transmission errors, the FL algorithm will converge to the optimal global FL model without any gaps. This result also corresponds to the result in the existing works (e.g., \cite{friedlander2012hybrid}). In the following section, we
show how one can leverage the result in Theorem \ref{th:1} to solve the proposed problem (\ref{eq:max}).

\section{Optimization of FL Training Loss}
In this section, our goal is to minimize the FL loss function when considering the underlying wireless network constraints. 
To solve the problem in (\ref{eq:max}), we must first simplify it.
From Theorem \ref{th:1}, we can see that, to minimize the training loss in (\ref{eq:max}), we need to only minimize the gap, $\frac  {2\zeta_1}  {LK}\sum\limits_{i = 1}^U K_i  \left(1-a_i+a_iq_i\left(\boldsymbol{r}_{i}, P_{i} \right)\right) \frac{1-A^t}{1-A}$. When $A\ge1$, the FL algorithm will not converge. In consequence, here, we only consider the minimization of the FL training loss when $A < 1$. Hence, as $t$ is large enough, which captures the asymptotic convergence behavior of FL, we have $A^t=0$. The gap can be rewritten as
 \begin{equation}\label{eq:gap}
 \begin{split}
&\frac  {2\zeta_1}  {LK}\sum\limits_{i = 1}^U K_i  \left(1-a_i+a_iq_i\left(\boldsymbol{r}_{i}, P_{i} \right)\right) \frac{1-A^t}{1-A}=\frac{\frac  {2\zeta_1}  {LK}\sum\limits_{i = 1}^U K_i  \left(1-a_i+a_iq_i\left(\boldsymbol{r}_{i}, P_{i} \right)\right) }{\frac{\mu}{L}-\frac  {4\mu\zeta_2} { {LK}}\sum\limits_{i = 1}^U K_i \left(1-a_i+a_iq_i\left(\boldsymbol{r}_{i}, P_{i} \right)\right) }.
\end{split}
  \end{equation}
  From (\ref{eq:gap}), we can observe that minimizing $\frac  {2\zeta_1}  {LK}\sum\limits_{i = 1}^U K_i  \left(1-a_i+a_iq_i\left(\boldsymbol{r}_{i}, P_{i} \right)\right) \frac{1-A^t}{1-A}$ only requires minimizing $\sum\limits_{i = 1}^U K_i  \left(1-a_i+a_iq_i\left(\boldsymbol{r}_{i}, P_{i} \right)\right) $. Meanwhile, since $a_i= \sum\limits_{n = 1}^R r_{i,n}$ and $q_i\left(\boldsymbol{r}_{i}, P_{i} \right)= \sum\limits_{n = 1}^R r_{i,n}q_{i,n}$, we have {\color{black}$q_i\left(\boldsymbol{r}_{i}, P_{i} \right)\le1$}, when $a_i=1$, and $q_i\left(\boldsymbol{r}_{i}, P_{i} \right)=0$, if $a_i=0$. In consequence, we have $a_iq_i\left(\boldsymbol{r}_{i}, P_{i} \right)=q_i\left(\boldsymbol{r}_{i}, P_{i} \right)$.
The problem in (\ref{eq:max}) can be simplified as
  \addtocounter{equation}{0}
\begin{equation}\label{eq:max1}
\begin{split}
\mathop {\min }\limits_{\boldsymbol{P}, \boldsymbol{R}}\sum\limits_{i = 1}^U K_i  \left(1-\sum\limits_{n = 1}^R r_{i,n}+q_i\left(\boldsymbol{r}_{i}, P_{i} \right)\right),
\end{split}
\end{equation}
\vspace{-0.3cm}
\begin{align}\label{c1}
\setlength{\abovedisplayskip}{-20 pt}
\setlength{\belowdisplayskip}{-20 pt}
&\!\!\!\!\!\!\!\!\rm{s.\;t.}\;\scalebox{1}{(\ref{eq:max}c)~--~(\ref{eq:max}f)}, \nonumber\\
&\scalebox{1}{$\;\;\;  r_{i,n} \in \left\{0,1\right\}, \;\;\;\;\;\forall i \in \mathcal{U}, n=1,\ldots, R,$}\tag{\theequation a}\\
&\scalebox{1}{$\;\;\; \sum\limits_{n = 1}^R r_{i,n} \le 1,\;\;\forall i \in \mathcal{U}. $} \tag{\theequation b}
\end{align} 
Next, we first find the optimal transmit power for each user given the uplink RB allocation matrix $\boldsymbol{R}$. Then, we find the uplink RB allocation to minimize the FL loss function. {\color{black}Since there always exist $\zeta_1$ and $\zeta_2$ that satisfy the constraint in (\ref{eq:pr1}) and $\|\nabla {f\left( { \boldsymbol{g}_{t} ,{\boldsymbol{x}_{ik}},{{y}_{ik}}} \right)}\|^2\leq \zeta_1+\zeta_2 \| \nabla F \left(\boldsymbol{g}_{t} \right)\|^2$, we do not add constraint (\ref{eq:pr1}) into (\ref{eq:max1}).}
\subsection{Optimal Transmit Power}

The optimal transmit power of each user $i$ can be determined by the following proposition.
  \begin{proposition}\label{theorem2}
\emph{Given the uplink RB allocation vector $\boldsymbol{r}_i$ of each user $i$, the optimal transmit power of each user $i$, $P_i^*$ is given by
\begin{equation}\label{eq:theorem2}
{\color{black}P_i^*\left(\boldsymbol{r}_i \right)=\min\left\{P_{\max}, P_{i,\gamma_\textrm{E}}\right\}},
\end{equation}
where $P_{i, \gamma_\textrm{E}}$ satisfies the equality $\varsigma \omega_i \vartheta^2 Z\left( \boldsymbol{w}_{i}\right)+\frac{P_{i, \gamma_\textrm{E}}Z\left(\boldsymbol{w}_i\right) }{c_{i}^\textrm{U} \left(\boldsymbol{r}_{i}, P_{i, \gamma_\textrm{E}}\right)    }=\gamma_\textrm{E} $.}
\end{proposition} 
\begin{proof} See Appendix B.
\end{proof}

From Proposition \ref{theorem2}, we see that the optimal transmit power depends on the size of the local FL model $Z\left(\boldsymbol{w}_i\right)$ and the interference in each RB. In particular, as the size of the local FL model increases, each user must spend more energy for training FL model and, hence, the energy that can be used for data transmission decreases. In consequence, the training loss increases. {\color{black}Hereinafter, for simplicity, $P_i^*$ is short for $P_i^*\left(\boldsymbol{r}_i \right)$}.   

\subsection{Optimal Uplink Resource Block Allocation}
Based on Proposition \ref{theorem2} and (\ref{eq:per}), the optimization problem in (\ref{eq:max1}) can be simplified as follows
  \addtocounter{equation}{0}
\begin{equation}\label{eq:max2}
\begin{split}
\mathop {\min }\limits_{\boldsymbol{R}}\sum\limits_{i = 1}^U K_i  \left(1-\sum\limits_{n = 1}^R r_{i,n}+ \sum\limits_{n = 1}^R r_{i,n}q_{i,n} \right),
\end{split}
\end{equation}
\vspace{-0.3cm}
\begin{align}\label{c1}
\setlength{\abovedisplayskip}{-20 pt}
\setlength{\belowdisplayskip}{-20 pt}
&\!\!\!\!\!\!\!\!\rm{s.\;t.}\;\scalebox{1}{(\ref{eq:max1}a), (\ref{eq:max1}b), \textrm{and} (\ref{eq:max}e)},\nonumber\\ 
&\scalebox{1}{$\;\;\;   l_{i}^\textrm{U}\left(\boldsymbol{r}_{i}, P_{i}^*\right)+l_{i}^\textrm{D} \le \gamma_\textrm{T}  ,\;\;\forall i \in \mathcal{U}, $} \tag{\theequation a}\\
&\scalebox{1}{$\;\;\;  e_{i}\left(\boldsymbol{r}_{i},P_{i}^*\right)\le \gamma_\textrm{E} ,\;\;\forall i \in \mathcal{U}.$} \tag{\theequation b}
\end{align} 
Obviously, {\color{black}the objective function (\ref{eq:max2}) is linear, the constraints are non-linear, and the optimization variables are integers. Hence, problem (\ref{eq:max2}) can be solved by using bipartite matching algorithm \cite{mahdian2011online}}. Compared to traditional convex optimization algorithms, using bipartite matching to solve problem \eqref{eq:max2}  does not require computing the gradients of each variable nor dynamically adjusting the step size for convergence. 

To use a bipartite matching algorithm for solving problem (\ref{eq:max2}),
we first transform the optimization problem into a bipartite matching problem. We construct a bipartite graph $\mathcal{A}=\left(\mathcal{U}\times\mathcal{R}, \mathcal{E}  \right)$ where $\mathcal{R}$ is the set of RBs that can be allocated to each user, each vertex in $\mathcal{U}$ represents a user and each vertex in $\mathcal{R}$ represents an RB, and $\mathcal{E}$ is the set of edges that connect to the vertices from each set $\mathcal{U}$ and $\mathcal{R}$.  Let $\chi_{in}\in\mathcal{E}$ be the edge connecting vertex $i$ in $\mathcal{U}$ and vertex $n$ in $\mathcal{R}$ with $\chi_{in} \in \left\{0,1\right\}$, where $\chi_{in}=1$ indicates that RB $n$ is allocated to user $i$, otherwise, we have $\chi_{in}=0$. Let matching $\mathcal{T}$ be a subset of edges in $\mathcal{E}$, in which no two edges share a common vertex in $\mathcal{R}$, such that each RB $n$ can only be allocated to one user (constraint (\ref{eq:max}e) is satisfied). Nevertheless, in $\mathcal{T}$, all of the edges associated with a vertex $i \in\mathcal{U}$ will not share a common vertex $n \in \mathcal{R}$, such that each user $i$ can occupy only one RB (constraint (\ref{eq:max}b) is satisfied). The weight of edge $\chi_{in}$ is given by
\begin{equation} \label{eq:psi}
\psi_{in}\!=\!\left\{ {\begin{array}{*{20}{c}}
  {\!\!\!\!K_i\!\left(q_{i,n}\!-1\right)\!,  l_{i}^\textrm{U}\!\left({r}_{i,n}, P_{i}^*\right)\!+\!l_{i}^\textrm{D} \!\le\! \gamma_\textrm{T}\!~\textrm{and}~\!e_{i}\!\left({r}_{i,n},P_{i}^*\right)\!\le\! \!\gamma_\textrm{E} }, \\ 
  {0,\;\;\; \;\;\;\;\;\;\;\;\;\textrm{otherwise}.\;\;\;\;\;\;\;\;\;\;\;\;\;\;\;\;\;\;\;\;\;\;\;\;\;\;\;\;\;\;\;\;\;\;\;\;\;\;\;\;\;\;\;\;\;\;\;\;\;} 
\end{array}} \right.
\end{equation}
From (\ref{eq:psi}), we can see that when RB $n$ is allocated to user $i$, if the delay and energy requirements cannot be satisfied, we will have $\psi_{in}=0$, which indicates that RB $n$ will not be allocated to user $i$. The goal of this formulated bipartite matching problem is to find an optimal matching set $\mathcal{T}^*$ that can minimize the weights of the edges in $\mathcal{T}^*$. A standard Hungarian algorithm \cite{jonker1986improving} can be used to find the optimal matching set $\mathcal{T}^*$. When the optimal matching set is found, the optimal RB allocation is determined. {\color{black} When the optimal RB allocation vector $\boldsymbol{r}_i^*$ is determined, the optimal transmit power of each device can be determined by (\ref{eq:theorem2})} and the optimal user selection can be determined by $a_i^*= \sum\limits_{n = 1}^R r_{i,n}^*$. Algorithm 1 summarizes the entire process of optimizing the user selection vector $\boldsymbol{a}$, RB allocation matrix $\boldsymbol{R}$, and the transmit power vector $\boldsymbol{P}$ for training the FL algorithm.

\begin{algorithm}[t]\footnotesize
{\color{black}
\caption{{\color{black}Proposed FL Over Wireless Networks}}   
\label{algorithm}   
\begin{algorithmic}[1] 
\vspace{1pt}  
\ENSURE Data rate of each user $c_{i}^\textrm{U} \left(\boldsymbol{r}_{i}, P_{i}\right)$ and $c_{i}^\textrm{D} $, the data size of local FL model, $Z\left(\boldsymbol{w}_i\right)$, packet error rate of each user $i$, $q_i\left(\boldsymbol{r}_{i}, P_{i}\right)$. \\ 
\vspace{1pt}
\STATE  Analyze the expected convergence of the federated learning based on (\ref{eq:theorem1}).    
\vspace{1pt}  
\STATE Find the optimal transmit power of each user over each RB using (\ref{eq:theorem2}).   
\vspace{1pt}  
\STATE Solve the optimization problem (\ref{eq:max2}) using a standard Hungarian algorithm and (24).  
\vspace{1pt}
 \STATE Implement the FL algorithm using optimal RB allocation matrix $\boldsymbol{R}^*$, user selection vector $\boldsymbol{a}^*$, and transmit power vector $\boldsymbol{P}^*$. 
\vspace{1pt}
\end{algorithmic}}
\end{algorithm}

 \subsection{Implementation and Complexity}
Next, we first analyze the implementation of the Hungarian algorithm. To implement the Hungarian algorithm for finding the optimal matching set $\mathcal{T}^*$, the BS must first calculate the packet error rate $q_{i,n}$, total delay $l_{i}^\textrm{U}\left({r}_{i,n}, P_{i}^*\right)+l_{i}^\textrm{D}$, and the energy consumption $e_{i}\left({r}_{i,n},P_{i}^*\right)$ of each user transmitting the local FL model over each RB $n$. To calculate the packet error rate $q_{i,n}$ and total delay $l_{i}^\textrm{U}\left({r}_{i,n}, P_{i}^*\right)+l_{i}^\textrm{D}$, the BS must know the SINR over each RB and the data size of FL models. The BS can use channel estimation methods to learn the SINR over each RB.
The data size of the FL model depends on the learning task. To implement an FL mechanism, the BS must first send the FL model information and the learning task information to the users. In consequence, the BS will learn the data size of FL model before the execution of the FL algorithm. To calculate the energy consumption $e_{i}\left({r}_{i,n},P_{i}^*\right)$ of each user, the BS must learn each user's device information such as CPU. This device information can be learned by the BS when the users initially connect to the BS. Given the packer error rate $q_{i,n}$, total delay $l_{i}^\textrm{U}\left({r}_{i,n}, P_{i}^*\right)+l_{i}^\textrm{D}$, and the energy consumption $e_{i}\left({r}_{i,n},P_{i}^*\right)$ of each user, the BS can compute $\psi_{in}$ according to (\ref{eq:psi}). Given $\psi_{in}$, the Hungarian algorithm can be used to find the optimal matching set $\mathcal{T}^*$. Since the optimization function in (\ref{eq:max2}) is linear, it admits an optimal matching set $\mathcal{T}^*$ and the Hungarian algorithm will finally find the optimal matching set $\mathcal{T}^*$.

With regards to the complexity of the Hungarian algorithm, it must first use $UR$ iterations to calculate the packer error rate, total delay, and energy consumption of each user over each RB.  
After that, the Hungarian algorithm will update the values of $\psi_{in}$ so as to find the optimal matching set $\mathcal{T}^*$. The worst complexity of the hungarian algorithm to find the optimal matching set $\mathcal{T}^*$ is $\mathcal O\left(U^2R\right)$ \cite{HungarianMaximumMatching}. In contrast, the best complexity is $\mathcal O\left(UR\right)$.
In consequence, the major complexity lies in calculating the weight of each edge and updating the edges in the matching set $\mathcal{T}$. However, in the Hungarian algorithm, we need to only perform simple operations such as $ K_i\left(q_{i,n}-1\right)$ without calculation for the gradients of each valuables nor  adjusting the step sizes as done in the optimization algorithms.  
Meanwhile, the Hungarian algorithm is implemented by the BS in a centralized manner and the BS will have sufficient computational resources to implement it.

\section{Simulation Results and Analysis}
For our simulations, we consider a circular network area having a radius $r=500$ m with one BS at its center servicing $U = 15$ uniformly distributed users. The other parameters used in simulations are listed in Table~I. {\color{black} The FL algorithm is simulated by using the Matlab Machine Learning Toolbox for linear regression and handwritten digit identification. For linear regression, each user implements a feedforward neural network (FNN) that consists of 20 neurons. The data used to train the FL algorithm is generated randomly from $\left[0,1\right]$. The input $x$ and the output $y$ follow the function $y=-2x+1+n\times0.4$ where $n $ follows a Gaussian distribution $ \mathcal{N}\left(0,1\right)$. {\color{black}The loss function is mean squared normalized error. For handwritten digit identification, each user trains an FNN that consists of 50 neurons using the MNIST dataset \cite{MNIST}. The loss function is cross entropy loss.  
  For comparison purposes, we use three baselines: a) an FL algorithm that optimizes user selection with random resource allocation, b) an FL algorithm that randomly determines user selection and resource allocation, which can be seen as a standard FL algorithm (e.g., similar to the one in \cite{konevcny2016federated}) that is not wireless-aware, and c) a wireless optimization algorithm that minimizes the sum packet error rates of all users via optimizing user selection, transmit power while ignoring FL parameters.} Our code is available at: \url{https://github.com/mzchen0/Wireless-FL}.


\begin{table}\footnotesize
  \newcommand{\tabincell}[2]{\begin{tabular}{@{}#1@{}}#2\end{tabular}}
\renewcommand\arraystretch{1}
 \caption{
    \vspace*{-0.05em}SYSTEM PARAMETERS}\vspace*{-0.6em}
\centering  
\begin{tabular}{|c|c|c|c|}
\hline
\textbf{Parameter} & \textbf{Value} & \textbf{Parameter} & \textbf{Value} \\
\hline
$\alpha$&2  &$N_0$& -174 dBm/Hz \\
\hline
 $P_B$ & 1 W & $B^\textrm{D}$& 20 MHz  \\
\hline
{\color{black}$m$} &{\color{black} 0.023 dB} &$B^\textrm{U}$& 1 MHz \\
\hline
$ \sigma_i $ &1& $P_{\max}$ & 0.01 W \\
\hline
$\vartheta$&$10^9$& $K_i$ & [12,10,8,4,2]  \\
\hline
$\varsigma $&$10^{-27} $& $\gamma_\textrm{T}$ & 500 ms\\
\hline
$\omega_i$ & 40&$\gamma_\textrm{E}$ & 0.003 J  \\
\hline
\end{tabular}
\end{table}

\subsection{FL for Linear Regression}
\begin{figure}[!t]
  \begin{center}
   \vspace{0cm}
    \includegraphics[width=10cm]{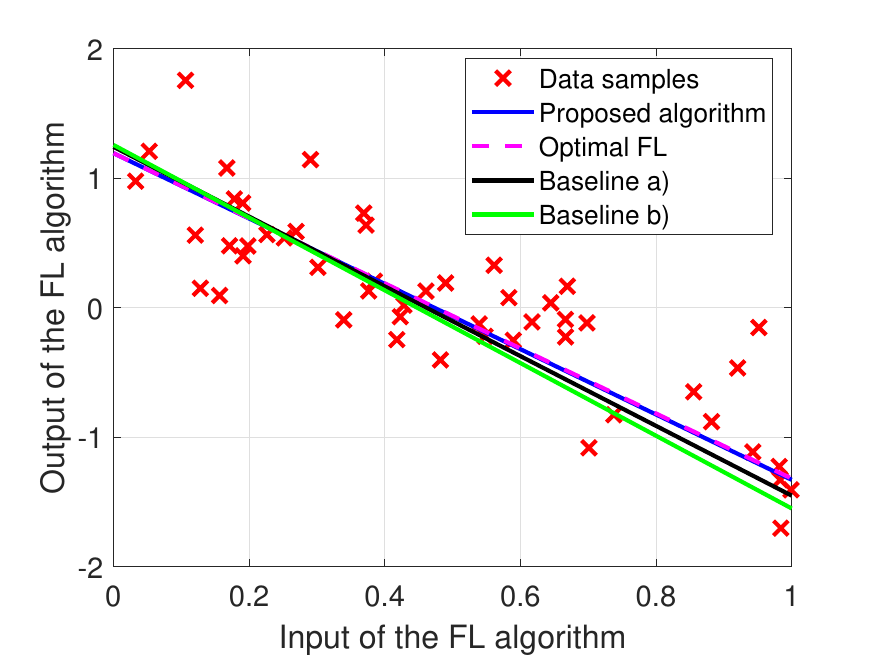}
    \vspace{-0.25cm}
    \caption{\label{fig6} An example of implementing FL for linear regression.}
  \end{center}\vspace{-0.6cm}
\end{figure}
Fig. \ref{fig6} shows an example of using FL for linear regression. In this figure, the red crosses are the data samples. In the optimal FL, the optimal RB allocation, user association, and transmit power powers are derived using a heuristic search method. From Fig. \ref{fig6}, we see that the proposed FL algorithm can fit the data samples more accurately than baselines a) and b). This is due to the fact that the proposed FL algorithm jointly considers the learning and wireless factors and, hence, it can optimize user selection and resource allocation to reduce the effect of wireless transmission errors on training FL algorithm and improve the performance of the FL algorithm. Fig. \ref{fig6} also shows that the proposed algorithm can reach the same performance as the optimal FL, which verifies that the proposed algorithm can find an optimal solution using the Hungarian algorithm.

\begin{figure}[!t]
  \begin{center}
   \vspace{0cm}
    \includegraphics[width=10cm]{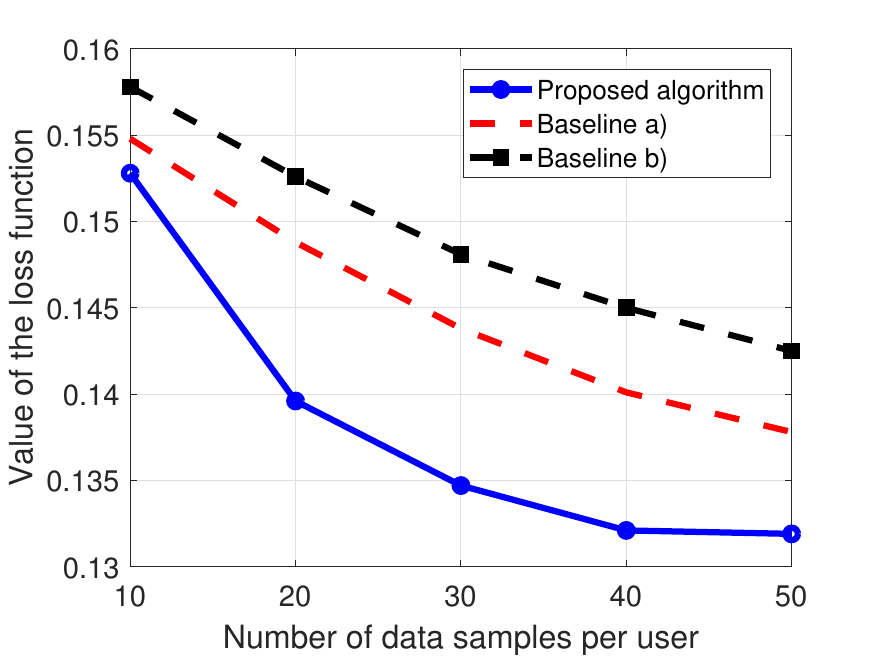}
    \vspace{-0.25cm}
    \caption{\label{fig7} Training loss as the number of data samples per user varies.}
  \end{center}\vspace{-0.6cm}
\end{figure}     

Fig. \ref{fig7} shows how the training loss changes as the number of data samples of each user varies.
From this figure, we observe that, as the number of data samples of each user increases, the values of the FL loss function of all of considered FL algorithms decrease. This is due to the fact that, as the number of data samples increases, all of the considered learning algorithms can use more data samples for training. Fig. \ref{fig7} also demonstrates that, when the number of data samples is less than 30, the training loss decreases quickly. However, as the number of data samples continues to increase, the training loss remains unchanged. This is due to the fact that as the number of data samples is over 30, the BS has enough data samples to approximate the gradient of the loss function.

%

\begin{figure}[!t]
  \begin{center}
   \vspace{0cm}
    \includegraphics[width=10cm]{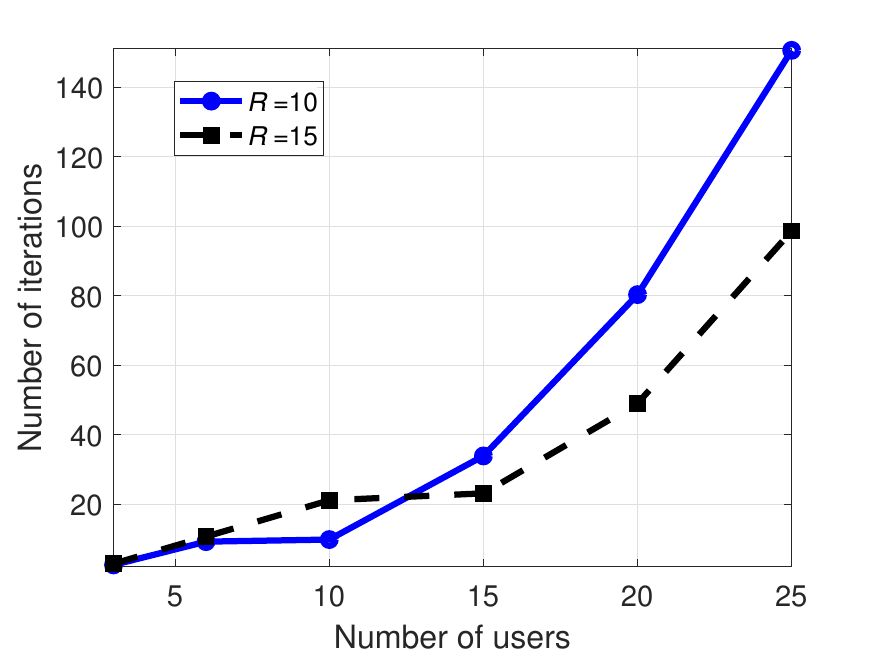}
    \vspace{-0.25cm}
 {\color{black}   \caption{\label{fig5} Number of iterations as the number of users varies.}}
  \end{center}\vspace{-0.6cm}
\end{figure}   

In Fig. \ref{fig5}, we show the number of iterations that the Hungarian algorithm needs to find the optimal RB allocation as a function of the number of users. From this figure, we can see that, as the number of users increases, the number of iterations needed to find the optimal RB allocation increases. This is because, as the number of users increases, the size of the edge weight matrix in (\ref{eq:psi}) increases and, hence, the Hungarian algorithm needs to use more iterations to find the optimal RB allocation. Fig. \ref{fig5} also shows when the number of users is smaller than the number of RBs, the number of iterations needed to find the optimal RB allocation increases slowly. However, as the number of users continues to increase, the number of iterations significantly increases. Fig. \ref{fig5} also shows that, when the number of users is larger than 10, the number of iterations needed to find the optimal RB allocation for a network with 10 RBs is larger than that of a network with 15 RBs. This is due to the fact that as the number of users is larger than 10, the gap between the number of users and the number of RBs for a network with 10 RBs is larger than that for a network with 15 RBs.

\begin{figure}[!t]
  \begin{center}
   \vspace{0cm}
    \includegraphics[width=8cm]{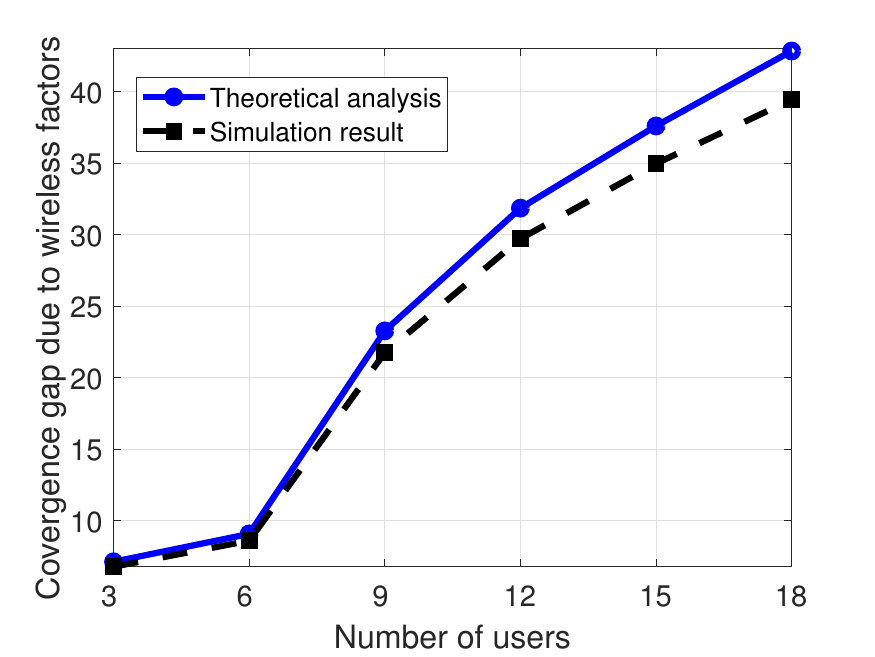}
    \vspace{-0.25cm}
 {\color{black}   \caption{\label{fig10} Convergence gap caused by wireless factors as the number of users changes.}}
  \end{center}\vspace{-0.6cm}
\end{figure}

{\color{black}Fig. \ref{fig10} shows how the convergence gap changes as the number of iterations changes. In Fig. \ref{fig10}, y-axis is the value of $\mathbb E \left( F\left(\boldsymbol{g}_{t+1}\right)-F\left(\boldsymbol{g}^{*}\right)\right)$ when the FL algorithm reaches convergence. That is, $\frac  {\zeta_1}  {2LK}\sum\limits_{i = 1}^U K_i  \left(1-a_i+a_iq_i\left(\boldsymbol{r}_{i}, P_{i} \right)\right) \frac{1}{1-A}$.
From Fig. \ref{fig10}, we can see that, the theoretical analysis derived in Theorem 1 is aligned with the simulation results with less than 9\% difference, thus corroborating the validity of Theorem 1.
From Fig. \ref{fig10}, we can also see that, as the number of users increases, the value of $\mathbb E \left( F\left(\boldsymbol{g}_{t+1}\right)-F\left(\boldsymbol{g}^{*}\right)\right)$ increases. This is because, as the number of users increases, the probability that the users cannot perform FL algorithm increases.}

%

\subsection{FL for Handwritten Digit Identification}

\begin{figure}[!t]
  \begin{center}
   \vspace{0cm}
    \includegraphics[width=10cm]{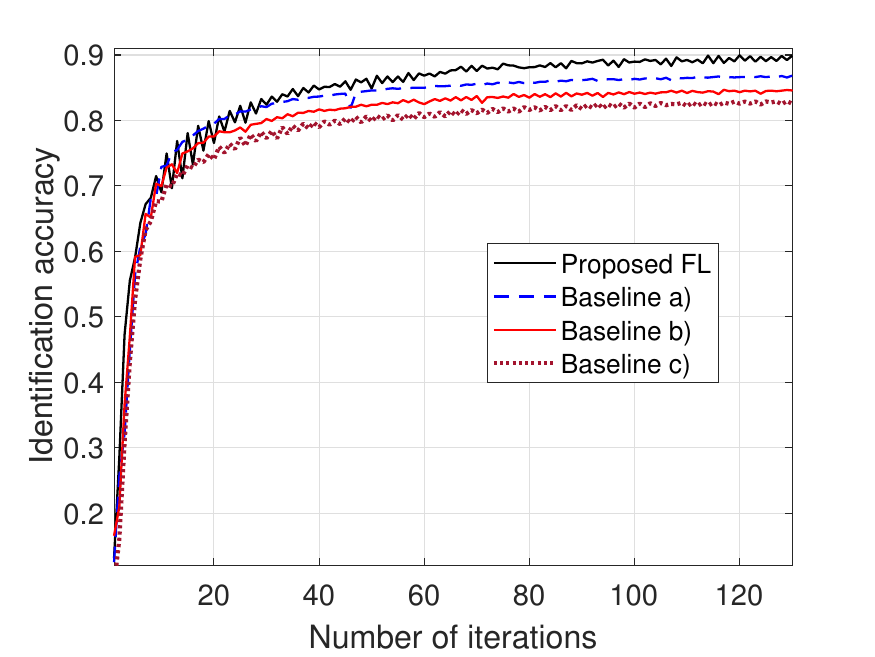}
    \vspace{-0.25cm}
   {\color{black} \caption{\label{fig4} Identification accuracy as the number of iterations varies.}}
  \end{center}\vspace{-0.6cm}
\end{figure}      

In Fig. \ref{fig4}, we show how the identification accuracy changes as the number of iterations varies.
From Fig. \ref{fig4}, we see that, as the number of iterations increases, the identification accuracy of all considered learning algorithms decreases first and, then remains unchanged. The fact that the identification accuracy remains unchanged demonstrates that the FL algorithm converges. From Fig. \ref{fig4}, we can also see that the increase speed in the value of identification accuracy is different during each iteration. This is due to the fact that the local FL models that are received by the BS may contain data errors and the BS may not be able to use them for the update of the global FL model. In consequence, at each iteration, the number of local FL models that can be used for the update of the global FL model will be different. Fig. \ref{fig4} also shows that a gap exists between the proposed algorithm and baselines a), b), and c). This gap is caused by the packet errors. 

\begin{figure}[!t]
  \begin{center}
   \vspace{0cm}
    \includegraphics[width=10cm]{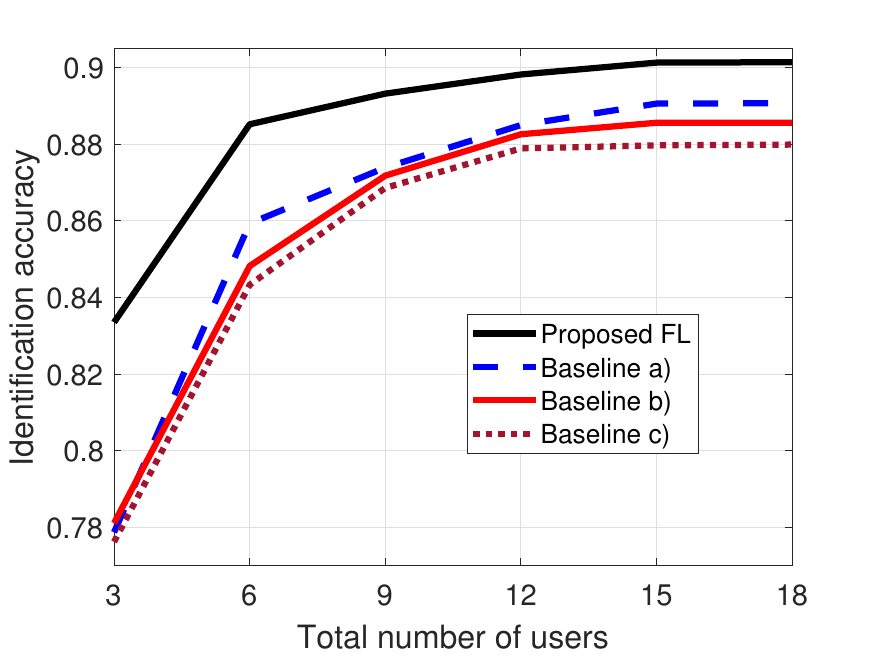}
    \vspace{-0.25cm}
 {\color{black}  \caption{\label{fig3} Identification accuracy as the total number of users varies ($R=12$).}}
  \end{center}\vspace{-0.6cm}
\end{figure}
{\color{black}Fig. \ref{fig3} shows how the identification accuracy changes as the total number of users varies. In this figure, an appropriate subset of users is selected to perform the FL algorithm. From Fig. \ref{fig3}, we can observe that, as the number of users increases, the identification accuracy increases. 
This is due to the fact that an increase in the number of users leads to more data available for the FL algorithm training and, hence, improving the accuracy of approximation of the gradient of the loss function. Fig. \ref{fig3} also shows that the proposed algorithm improves the identification accuracy by, respectively, up to {\color{black} 1.2\%, 1.7\%, and 2.3\%} compared to baselines a), b) and c) as the network consists of 18 users. 
The {\color{black} 1.2\%} improvement stems
from the fact that the proposed algorithm optimizes the resource allocation.  
The {\color{black}1.7\%} improvement stems from the fact
that the proposed algorithm joint considers learning and wireless effects and, hence, it can optimize the user selection and resource allocation to reduce the FL loss function. {\color{black} The 2.3\% improvement stems from the fact that the proposed algorithm optimizes wireless factors while considering FL parameters such as the number of training data samples.} 
Fig. \ref{fig3} also shows that when the number of users is less than 12, the value of the identification accuracy increases quickly. In contrast, as the number of users continues to increase, the identification accuracy increases slowly. This is because, for a higher number of users,
the BS will have enough data samples to accurately approximate the gradient of the loss function.} 

\begin{figure}[!t]
  \begin{center}
   \vspace{0cm}
    \includegraphics[width=10cm]{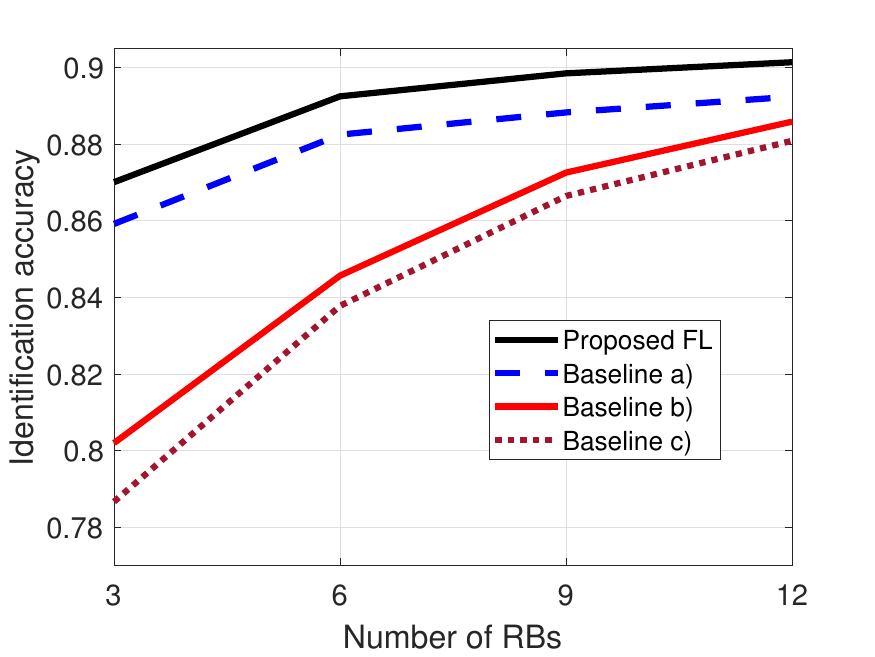}
    \vspace{-0.25cm}
    {\color{black}
   {\color{black}  \caption{\label{fig9} Identification accuracy changes as the number of RBs varies ($U=15$).}}}
  \end{center}\vspace{-0.6cm}
\end{figure}

{\color{black} Fig. \ref{fig9} shows how the identification accuracy changes as the number of RBs changes. From Fig. \ref{fig9} of the response, we can see that, as the number of RBs increases, the identification accuracy resulting from all of the considered FL algorithms increases. This is due to the fact that, as the number of RBs increases, the number of users that can perform the FL algorithm increases. From this figure, we can also see that, the proposed FL algorithm can achieve up to {\color{black} 1.4\%, 3.5\%, and 4.1\%} gains in terms of the identification accuracy compared to baselines a), b), and c) for a network with 9 RBs. This is because the proposed FL algorithm can optimize the RB allocation, transmit power, and user selection and thus minimizing the loss function values. }

\begin{figure}[!t]
  \begin{center}
   \vspace{0cm}
    \includegraphics[width=13cm]{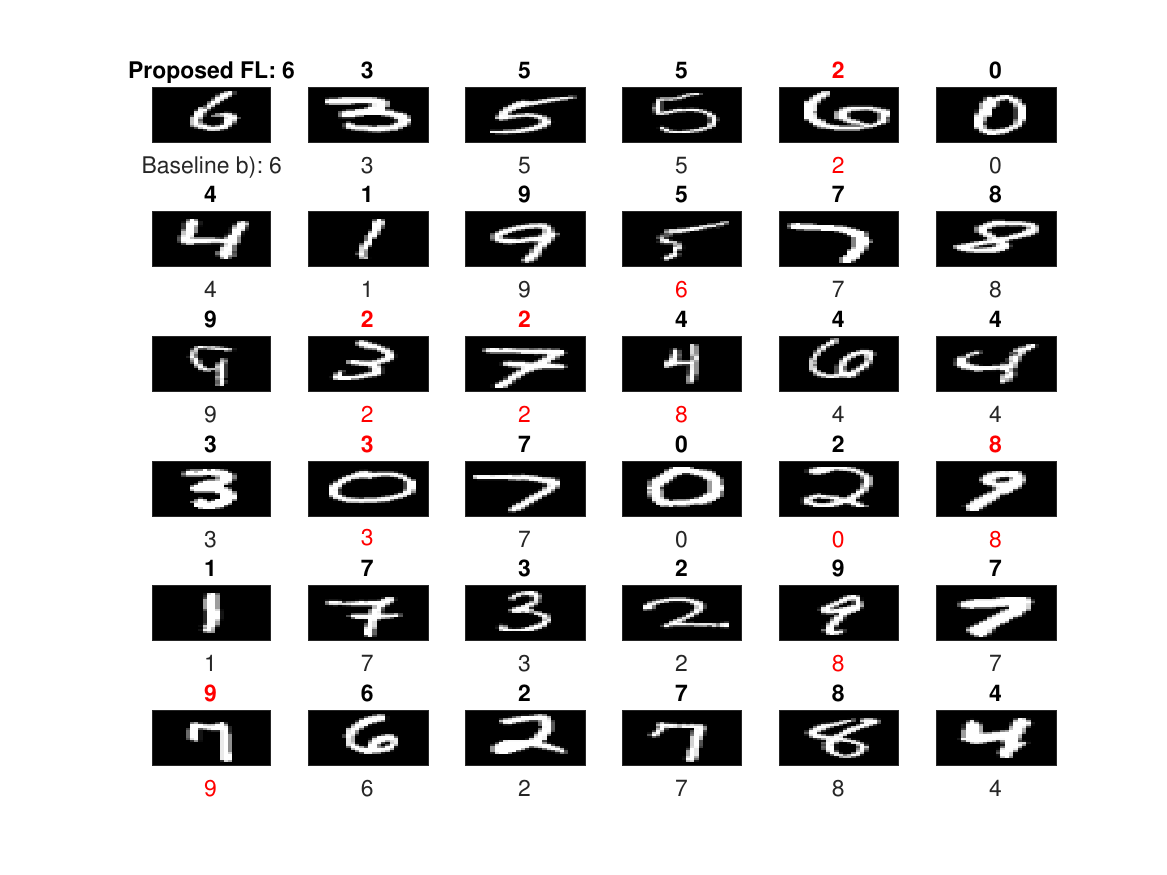}
    \vspace{-0.8cm}
{\color{black}    \caption{\label{figCNN} An example of implementing FL for handwritten digit identification.}}
  \end{center}\vspace{-0.6cm}
\end{figure}
{\color{black} Fig. \ref{figCNN} shows  one example of implementing the proposed FL algorithm for handwritten digit identification. In particular, each user trains a convolutional neural network (CNN) using the MNIST dataset. In this simulation, CNNs are generated by the Matlab Machine Learning Toolbox and each user has 2000 training data samples to train the CNN. {\color{black}From this figure, we can see that, for 36 handwritten digit identification, the proposed algorithm correctly identify 30 handwritten digits while baseline b) correctly identify 27 handwritten digits.} Hence, the proposed FL algorithm can more accurately identify the handwritten digits than baseline b). This is because the proposed FL algorithm can minimize the packet error rate of the users and hence improving the FL performance. {\color{black}From Fig. \ref{figCNN}, we can see that, even though CNNs are not convex, the proposed FL algorithm can also improve the FL performance. This is another evidence to show that our assumptions can still apply to practical FL solutions. }}

\section{Conclusion}
In this paper, we have developed a novel framework that enables the implementation of FL algorithms over wireless networks. We have formulated an optimization problem that jointly considers user selection and resource allocation for the minimization of FL training loss. To solve this problem, we have derived a closed-form expression for the expected convergence rate of the FL algorithm that considers the limitations of the wireless medium. Based on the derived expected convergence rate, the optimal transmit power is determined given the user selection and uplink RB allocation. Then, the Hungarian algorithm is used to find the optimal user selection and RB allocation so as to minimize the FL loss function.  
Simulation results have shown that the joint federated learning and communication framework yields significant improvements in the performance compared to existing implementations of the FL algorithm that does not account for the properties of the wireless channel.

  \section*{Appendix}
\subsection{Proof of Theorem \ref{th:1}}\label{Ap:a}
To prove Theorem 1, we first rewrite $F \left(\boldsymbol{g}_{t+1} \right)$ using the second-order Taylor expansion, which can be expressed as
\begin{equation}\label{itproofas3}
\begin{split}
F\left(\boldsymbol{g}_{t+1} \right) 
=&F \left(\boldsymbol{g}_{t} \right)
+\left( \boldsymbol{g}_{t+1} - \boldsymbol{g}_{t}\right)^{T}  \nabla F \left(\boldsymbol{g}_{t} \right)+\frac  {1} 2 \left( \boldsymbol{g}_{t+1} - \boldsymbol{g}_{t}\right)^T
\nabla^2 F \left(\boldsymbol{g} \right)
\left( \boldsymbol{g}_{t+1} - \boldsymbol{g}_{t}\right), \\
\leq& F\left(\boldsymbol{g}_{t} \right)
+\left( \boldsymbol{g}_{t+1} - \boldsymbol{g}_{t}\right)^{T}  \nabla F\left(\boldsymbol{g}_{t} \right)
+\frac  {L } 2 \| \boldsymbol{g}_{t+1} - \boldsymbol{g}_{t}\|^2,
\end{split}
\end{equation}
where the inequality stems from the assumption in (\ref{itproofas2_2}). 
 Given the learning rate $\lambda=\frac{1}{L}$, based on (\ref{itproofeq1}), the expected optimization function $\mathbb E \left( F\left(\boldsymbol{g}_{t+1}\right)\right)$ can be expressed as
 \begin{equation}\label{itproofas3_2}
\begin{split}
\mathbb E \left( F\left(\boldsymbol{g}_{t+1}\right)\right)
\leq 
&\mathbb E\bigg(F\left(\boldsymbol{g}_{t} \right)
-\lambda(  \nabla F\left(\boldsymbol{g}_{t} \right)-\boldsymbol{o})^{T}  \nabla F\left(\boldsymbol{g}_{t} \right)+\frac  {L \lambda^2} {2} \| \nabla F \left(\boldsymbol{g}_{t} \right)-  \boldsymbol o  \|^2\bigg), \\
\mathop = \limits^{\left( a \right)} &\mathbb E\left(F \left(\boldsymbol{g}_{t} \right)\right)-\frac{1}{2L}\|  \nabla F\left(\boldsymbol{g}_{t}\right) \|^2 +\frac{1}{2L}\mathbb E\left( \|\boldsymbol{o}\|^2 \right),
\end{split}
\end{equation}
where (a) stems from the fact that $\frac  {L \lambda^2} {2} \| \nabla F (\boldsymbol{g}_{t} )- \boldsymbol{o} \|^2=\frac  {1} {2L} \| \nabla F (\boldsymbol{g}_{t} )\|^2-\frac{1}{L}\boldsymbol o^{T} \nabla F (\boldsymbol{g}_{t} )+\frac{1}{2L}\| \boldsymbol o\|^2$. Next, we derive $\mathbb E \left(\| \boldsymbol o \|^2\right)$, which can be given as follows
{\color{black}
\begin{equation}\label{eq:o}\small
\begin{split}
\mathbb E \left(\| \boldsymbol o \|^2\right)
&=\mathbb E \left(\left\|  \nabla F (\boldsymbol{g}_{t} )-\frac{\sum\limits_{i = 1}^U \sum\limits_{k = 1}^{K_i} a_i{\nabla f\left( {\boldsymbol{g},{\boldsymbol{x}_{ik}},{{y}_{ik}}} \right)}C\left( \boldsymbol{w}_i \right)}{{\sum\limits_{i = 1}^U K_ia_iC\left( \boldsymbol{w}_i \right)}}\right \|^2\right),\\
&\!\!\!\!\!\!\!\!\!\!\!\!\!\!\!\!\!\!=\mathbb E \left(\left\|-\frac{ \left(K-{{\sum\limits_{i = 1}^U K_ia_iC\left( \boldsymbol{w}_i \right)}}  \right)\sum\limits_{i \in \mathcal{N}_1} \sum\limits_{k = 1}^{K_i} {\nabla f\left( {\boldsymbol{g},{\boldsymbol{x}_{ik}},{{y}_{ik}}} \right)}}{{K{{\sum\limits_{i = 1}^U K_ia_iC\left( \boldsymbol{w}_i \right)}}}}+\frac{\sum\limits_{i \in \mathcal{N}_2} \sum\limits_{k = 1}^{K_i} {\nabla f\left( {\boldsymbol{g},{\boldsymbol{x}_{ik}},{{y}_{ik}}} \right)}}{K} \right\|^2\right),\\
&\!\!\!\!\!\!\!\!\!\!\!\!\!\!\!\!\!\!\le\mathbb E \left(\frac{ \left(K-{{\sum\limits_{i = 1}^U K_ia_iC\left( \boldsymbol{w}_i \right)}}  \right) \sum\limits_{i \in \mathcal{N}_1} \sum\limits_{k = 1}^{K_i} \left\| {\nabla f\left( {\boldsymbol{g},{\boldsymbol{x}_{ik}},{{y}_{ik}}} \right)}\right\|}{{K{{\sum\limits_{i = 1}^U K_ia_iC\left( \boldsymbol{w}_i \right)}}}}  +\frac{\sum\limits_{i \in \mathcal{N}_2} \sum\limits_{k = 1}^{K_i} \left\|{\nabla f\left( {\boldsymbol{g},{\boldsymbol{x}_{ik}},{{y}_{ik}}} \right)}\right\| }{K}\right)^2,
\end{split}
\end{equation}
where $\mathcal{N}_1=\left\{a_i=1, C\left( \boldsymbol{w}_i \right)=1| i \in \mathcal{U} \right\}$ is the set of users that correctly transmit their local
FL models to the BS and $\mathcal{N}_2=\left\{i \in \mathcal{U}| i \notin \mathcal{N}_1 \right\}$. The inequality equation in (\ref{eq:o}) is achieved by the triangle-inequality. Since $\|\nabla {f\left( { \boldsymbol{g}_{t} ,{\boldsymbol{x}_{ik}},{{y}_{ik}}} \right)}\|\leq \sqrt{\zeta_1+\zeta_2 \| \nabla F \left(\boldsymbol{g}_{t} \right)\|^2}$, we have $\sum\limits_{i \in \mathcal{N}_1} \sum\limits_{k = 1}^{K_i} \left\| {\nabla f\left( {\boldsymbol{g},{\boldsymbol{x}_{ik}},{{y}_{ik}}} \right)}\right\|\leq \sqrt{\zeta_1+\zeta_2 \| \nabla F \left(\boldsymbol{g}_{t} \right)\|^2}\sum\limits_{i = 1}^U K_ia_iC\left( \boldsymbol{w}_i \right)$ and $\sum\limits_{i \in \mathcal{N}_2} \sum\limits_{k = 1}^{K_i} \left\| {\nabla f\left( {\boldsymbol{g},{\boldsymbol{x}_{ik}},{{y}_{ik}}} \right)}\right\|\leq \sqrt{\zeta_1+\zeta_2 \| \\ \nabla F \left(\boldsymbol{g}_{t} \right)\|^2}\times\left(K-\sum\limits_{i = 1}^U K_ia_iC\left( \boldsymbol{w}_i \right)\right)$. Hence, $\mathbb E \left(\| \boldsymbol o \|^2\right)$ can be expressed by
\begin{equation}\label{eq:o1}
\begin{split}
\mathbb E \left(\| \boldsymbol o \|^2\right)
\le\frac{4}{{K^2}} {\mathbb E\left(K-{{\sum\limits_{i = 1}^U K_ia_iC\left( \boldsymbol{w}_i \right)}}  \right)^2 \left( {\zeta_1+\zeta_2 \| \nabla F \left(\boldsymbol{g}_{t} \right)\|^2} \right) }.
\end{split}
\end{equation}
Since $K \ge K- {{\sum\limits_{i = 1}^U K_ia_iC\left( \boldsymbol{w}_i \right)}}\ge0$, we have
\begin{equation}\label{eq:o2}
\begin{split}
\mathbb E \left(\| \boldsymbol o \|^2\right)
\le \frac{4}{{K}} {\mathbb E\left(K-{{\sum\limits_{i = 1}^U K_ia_iC\left( \boldsymbol{w}_i \right)}}  \right) 
\left( {\zeta_1+\zeta_2 \| \nabla F \left(\boldsymbol{g}_{t} \right)\|^2} \right) }. 
\end{split}
\end{equation}

Since $K={{\sum\limits_{i = 1}^U K_i}}$ and $\mathbb E \left(C\left( \boldsymbol{w}_i \right)\right)=1-q_i\left(\boldsymbol{r}_{i}, P_{i} \right)$, (\ref{eq:o2}) can be simplified as follows 
\begin{equation}\label{eq:o3}
\begin{split}
\mathbb E \left(\| \boldsymbol o \|^2\right)
&= { \frac{4}{{K}} \mathbb E\left({{\sum\limits_{i = 1}^U K_i \left(1-a_iC\left( \boldsymbol{w}_i \right)\right)}}  \right)\left( {\zeta_1+\zeta_2 \| \nabla F \left(\boldsymbol{g}_{t} \right)\|^2} \right) }, \\
&= { \frac{4}{{K}}{{\sum\limits_{i = 1}^U K_i\left(1-a_i+a_iq_i\left(\boldsymbol{r}_{i}, P_{i} \right)\right)}}\left( {\zeta_1+\zeta_2 \| \nabla F \left(\boldsymbol{g}_{t} \right)\|^2} \right) }.
\end{split}
\end{equation}
Substituting \eqref{eq:o3} into (\ref{itproofas3_2}), we have
{\color{black}\begin{equation}\label{eq:EF}
\begin{split}
\mathbb E \left( F\left(\boldsymbol{g}_{t+1}\right)\right) \leq&
\mathbb E\left(F \left(\boldsymbol{g}_{t} \right) \right)+\frac  {2\zeta_1}  {LK} \sum\limits_{i = 1}^U K_i  \left(1-a_i+a_iq_i\left(\boldsymbol{r}_{i}, P_{i} \right)\right)\\
&-\frac{1}{2L} \left(1- \frac{4\zeta_2}{K}\sum\limits_{i = 1}^U K_i \left(1-a_i+a_iq_i\left(\boldsymbol{r}_{i}, P_{i} \right)\right) \right)   \| \nabla F\left(\boldsymbol{g}_{t}\right)\|^2.
\end{split}
\end{equation}}
Subtract $\mathbb E \left( F\left(\boldsymbol{g}^{*} \right) \right)$ in both sides of (\ref{eq:EF}), we have
\begin{equation}\label{eq:EF2}
\begin{split}
\mathbb E \left( F\left(\boldsymbol{g}_{t+1}\right)-F\left(\boldsymbol{g}^{*}\right)\right)  \leq& \mathbb E\left(F \left(\boldsymbol{g}_{t}\right)-F\left(\boldsymbol{g}^{*}\right)\right)+\frac  {2\zeta_1}  {LK}\sum\limits_{i = 1}^U K_i  \left(1-a_i+a_iq_i\left(\boldsymbol{r}_{i}, P_{i} \right)\right)\\
 &-\frac{1}{2L} \left(1- \frac{4\zeta_2}{K}\sum\limits_{i = 1}^U K_i  \left(1-a_i+a_iq_i\left(\boldsymbol{r}_{i}, P_{i} \right)\right) \right)   \| \nabla F\left(\boldsymbol{g}_{t}\right)\|^2.
\end{split}
\end{equation}}

{\color{black}Given (\ref{itproofas2}) and (\ref{itproofas2_2}), we have \cite{boyd2004convex}} \begin{equation} \label{itproofeq5_3}
\| \nabla F\left(\boldsymbol{g}_{t}\right)\|^2 \geq 2 \mu \left( 
F \left(\boldsymbol{g}_{t} \right)- F \left(\boldsymbol{g}^* \right) \right).
\end{equation}
{\color{black}
Substituting \eqref{itproofeq5_3} into \eqref{eq:EF2}, we have
\begin{equation}\label{eq:EF3}
\begin{split}
\mathbb E \left( F\left(\boldsymbol{g}_{t+1}\right)-F\left(\boldsymbol{g}^{*}\right)\right)  \leq& \frac  {2\zeta_1}  {LK}\sum\limits_{i = 1}^U K_i  \left(1-a_i+a_iq_i\left(\boldsymbol{r}_{i}, P_{i} \right)\right)+A \mathbb E\left(F \left(\boldsymbol{g}_{t} \right)-F\left(\boldsymbol{g}^{*}\right)\right),
\end{split}
\end{equation}
where $A=1-\frac{\mu}{L}+ \frac  {4\mu\zeta_2} { {LK}}\sum\limits_{i = 1}^U K_i \left(1-a_i+a_iq_i\left(\boldsymbol{r}_{i}, P_{i} \right)\right)$. Applying (\ref{eq:EF3}) recursively, we have
\begin{equation}\label{eq:EF4}
\begin{split}
\mathbb E \left( F\left(\boldsymbol{g}_{t+1}\right)-F\left(\boldsymbol{g}^{*}\right)\right)
 \leq& \frac  {2\zeta_1}  {LK}\sum\limits_{i = 1}^U K_i  \left(1-a_i+a_iq_i\left(\boldsymbol{r}_{i}, P_{i} \right)\right)\sum\limits_{k=0}^{t-1} A^k+A^t  \mathbb E\left(F \left(\boldsymbol{g}_{0} \right)-F\left(\boldsymbol{g}^{*}\right)\right),\\
 =&\frac  {2\zeta_1}  {LK}\sum\limits_{i = 1}^U K_i  \left(1-a_i+a_iq_i\left(\boldsymbol{r}_{i}, P_{i} \right)\right) \frac{1-A^t}{1-A}+A^t  \mathbb E\left(F \left(\boldsymbol{g}_{0} \right)-F\left(\boldsymbol{g}^{*}\right)\right).
\end{split}
\end{equation}}
This completes the proof.

\subsection{Proof of Proposition \ref{theorem2}}\label{Ap:b}

To prove Proposition \ref{theorem2}, we first prove that $e_{i}\left( \boldsymbol{r}_{i},P_{i}\right)$ is an increasing function of $P_i$.
Based on \eqref{eq:uplinkdatarate} and \eqref{eq:energy}, we have
\begin{equation}
e_{i}\left( \boldsymbol{r}_{i},P_{i}\right) =\varsigma \omega_i \vartheta^2 Z\left( \boldsymbol{X}_{i}\right)+
\frac{P_{i}}{\sum\limits_{n = 1}^R r_{i,n} B^\textrm{U}{\log _2}\left( 1 +  \kappa_{i,n} P_{i}  \right)},
\end{equation}
where $\kappa_{i,n}={\frac{{ {h_{i}}}}{\sum\limits_{i' \in \mathcal{U}_n'}{P_{i'}}h_{i'}+B^\textrm{U}N_0}}$.
The first derivative of $e_{i}\left( \boldsymbol{r}_{i},P_{i}\right)$ with respect to $P_i$ is given by
\begin{equation}
\begin{split}
&\frac{\partial e_{i}\left( \boldsymbol{r}_{i},P_{i}\right)}{ \partial P_i}=\frac{\left(\ln2\right)\sum\limits_{n = 1}^R \frac{r_{i,n}}{1 +  \kappa_{i,n} P_{i} } \left({\left(1 +  \kappa_{i,n} P_{i} \right)\ln\left(1 +  \kappa_{i,n} P_{i} \right)- \kappa_{i,n} P_{i} }\right)}
{\left(\sum\limits_{n = 1}^R r_{i,n} B^\textrm{U}{\ln}\left( 1 +  \kappa_{i,n} P_{i}  \right)\right)^2}.
\end{split}
\end{equation}
Since {\color{black}$\frac{\partial e_{i}\left( \boldsymbol{r}_{i},P_{i}\right)}{ \partial P_i}$ is always positive when $P_i>0$}, $e_{i}\left( \boldsymbol{r}_{i},P_{i}\right)$ is a monotonically increasing function when $P_i>0$.
   Contradiction is used to prove Proposition \ref{theorem2}. We assume that $P'_i$ ($P'_i\ne P_i^* $) is the optimal transmit power of user $i$. In (\ref{eq:max}d), $e_{i}\left(\boldsymbol{r}_{i}^*,P_{i, \gamma_\textrm{E}}\right)$ is a monotonically increasing function of $P_i$. Hence, as  $P'_i > P_i^* $, $e_{i}\left(\boldsymbol{r}_{i}^*,P'_{i}\right)> \gamma_\textrm{E}$, which does not meet the constraint (\ref{eq:max}f). From (\ref{eq:per}), we see that, the packer error rates decrease as the transmit power increases. Thus, as $P'_i < P_i^* $, we have $ q_{i}\left(\boldsymbol{r}_i, P^*_i\right) \leqslant  q_{i}\left(\boldsymbol{r}_i, P'_i\right)$. In consequence, as $P'_i < P_i^* $, $P'_i $ cannot minimize the function in (\ref{eq:max1}). Hence, we have $P'_i =P_i^* $. This completes the proof.

\bibliographystyle{IEEEbib}
\bibliography{references1}
\end{document}